\RequirePackage{fix-cm}
\documentclass[smallextended]{svjour3}       
\usepackage[11pt]{extsizes}
\smartqed  
\usepackage{graphicx}

\usepackage{amssymb}
\usepackage{amsmath}

\usepackage{enumerate}
\usepackage{dsfont}
\usepackage[colorlinks, citecolor=red]{hyperref}
\usepackage{comment,cite,color}
\usepackage{cite,color}

\usepackage{mathrsfs}
\usepackage{epsfig}
\usepackage{lscape}
\usepackage{subfigure}

\usepackage{caption}
\usepackage{bm}
\usepackage{algorithm}
\usepackage{algpseudocode}
\usepackage[english]{babel}

\newcommand{\xkh}[1]{\left(#1\right)}
\newcommand{\dkh}[1]{\left\{#1\right\}}

\newcommand{\nj}[1]{\langle {#1} \rangle}

\newcommand{\norm}[1]{\|{#1}\|_2}

\newcommand{\norms}[1]{\|{#1}\|}
\newcommand{\abs}[1]{\lvert#1\rvert}
\newcommand{\Abs}[1]{\left\lvert#1\right\rvert}

\newcommand{\E}{{\mathbb E}}

\newcommand{\R}{{\mathbb {R}}}
\newcommand{\Rn}{{\mathbb R}^n}

\newcommand{\T}{\top}

\newcommand{\vx}{{\bm x}}
\newcommand{\vy}{{\bm y}}
\newcommand{\vu}{{\bm u}}
\newcommand{\vv}{{\bm v}}
\newcommand{\vw}{{\bm w}}
\newcommand{\vz}{{\bm z}}

\newcommand{\hz}{{\hat{\bm{z}}}}

\newcommand{\vb}{{\bm b}}

\newcommand{\va}{{\bm a}}

\newcommand{\RNum}[1]{\uppercase\expandafter{\romannumeral #1\relax}}

\renewenvironment{proof}{\noindent {\it \textbf{Proof}~}}{\hfill \qed \par}

\begin{document}

\title{Nearly optimal bounds for the global geometric landscape of phase retrieval \thanks{J. F. Cai was supported in part by Hong Kong Research Grant Council grants 16309518, 16309219, 16310620, 16306821.
Y. Wang was supported in part by the Hong Kong Research Grant Council grants 16306415 and 16308518. }
}


\author{Jian-Feng Cai \and  Meng Huang    \and  Dong Li \and Yang Wang     
}


\institute{J. F. Cai \at
              Department of Mathematics, The Hong Kong University of Science and Technology, Hong Kong, China \\
              \email{jfcai@ust.hk}           
           \and
           M. Huang \at
              School of Mathematical Sciences, Beihang University, Beijing, 100191, China\\
              \email{menghuang@buaa.edu.cn}
           \and
           D. Li \at
              SUSTech International Center for Mathematics and Department of Mathematics, Southern University of Science and Technology, Shenzhen, China\\
             \email{lid@sustech.edu.cn}     
             \and
              Y. Wang \at
              Department of Mathematics, The Hong Kong University of Science and Technology, Hong Kong, China \\
              \email{yangwang@ust.hk}                      
}


\maketitle

\begin{abstract}
The phase retrieval problem is concerned with  recovering an unknown signal $\vx\in \Rn$ from a set of  magnitude-only measurements $y_j=\abs{\nj{\va_j,\vx}}, \; j=1,\ldots,m$. A natural least squares formulation can be used to solve this problem efficiently even with random initialization, despite its non-convexity of the loss function. One way to explain this surprising phenomenon is the benign geometric landscape:  (1) all local minimizers are global; and (2) the objective function has a negative curvature around each saddle point and local maximizer. In this paper, we show that  $m=O(n \log n)$ Gaussian random measurements are sufficient to guarantee the loss function of a commonly used estimator has such benign geometric landscape with high probability. This is a step toward answering the open problem given by Sun-Qu-Wright  \cite{Sun18}, in which  the authors suggest that $O(n \log n)$ or even $O(n)$ is enough to guarantee the favorable geometric property.
\end{abstract}

\section{Introduction}
\subsection{Background}
In a prototypical  phase retrieval problem, one is interested in  how to recover an unknown signal $\vx\in \Rn$ from a series of magnitude-only measurements 
\begin{equation} \label{eq:meayi}
y_j=\abs{\nj{\va_j,\vx}}, ~j=1,\ldots,m,
\end{equation}
where $\va_j \in \Rn, j=1,\ldots,m$ are given vectors and $m$ is the number of measurements. This problem is of fundamental importance in numerous areas of physics and engineering such as X-ray crystallography \cite{harrison1993phase,millane1990phase}, microscopy
\cite{miao2008extending}, astronomy \cite{fienup1987phase}, coherent diffractive
imaging \cite{shechtman2015phase,Gerchberg1972} and optics
\cite{walther1963question} etc, where the optical detectors  can only record the magnitude of signals while losing the phase information. Despite its simple mathematical formulation, it has been
shown that reconstructing a finite-dimensional discrete signal from the magnitude of its Fourier transform is generally an {\em NP-complete} problem \cite{Sahinoglou}.

Due to the practical ubiquity of the phase retrieval problem, many algorithms have been designed for this problem.  For example, based on the technique ``matrix-lifting'', the phase retrieval problem can be recast as a low rank matrix recovery problem.
By  using convex relaxation one can show that the matrix recovery problem under suitable conditions is equivalent to a convex optimization problem  \cite{phaselift,Phaseliftn,Waldspurger2015}.  However, since the matrix-lifting technique involves semidefinite
program for $n\times n$ matrices, the computational cost is prohibitive for large scale problems. 
In contrast, many non-convex algorithms bypass the lifting step and operate directly on the lower-dimensional ambient space, making them much more computationally efficient. Early non-convex algorithms were mostly based on the technique of alternating projections, e.g. Gerchberg-Saxton \cite{Gerchberg1972} and Fineup \cite{ER3}. The main drawback, however, is the lack of theoretical guarantee. Later Netrapalli et al \cite{AltMin} proposed the AltMinPhase  algorithm based on a technique known as {\em spectral initialization}. They proved that the algorithm linearly converges to the true solution with $O(n \log^3 n)$ resampling Gaussian random measurements. This  work led 
further to several other non-convex algorithms based on spectral initialization. 
A common thread  is  first choosing a good initial guess through spectral initialization, and then solving an optimization model through gradient descent, such as \cite{WF,TWF,Gaoxu,TAF,RWF,huangwang,tan2019phase}. We refer the reader to survey papers  \cite{shechtman2015phase,Chinonconvex,jaganathan2016phase} for accounts of recent developments in the theory, algorithms and applications of phase retrieval.

\subsection{Prior arts and motivation}
As stated earlier,  producing a good initial guess using carefully-designed initialization seems to be a prerequisite for prototypical  non-convex algorithms to succeed with  good theoretical guarantee. A natural and fundamental  question is:

{\em Is it possible for non-convex algorithms to achieve successful recovery with a random initialization }?  

Recently, Ju Sun et al.  carried out a deep study of the global geometric structure of the loss function:
\begin{equation}\label{eq:mod1}
F(\vz)=\frac{1}{m}\sum_{j=1}^m \xkh{\abs{\nj{\va_j,\vz}}^2-y_j^2}^2,
\end{equation}
where $y_j$ are measurements given in \eqref{eq:meayi}. They proved that the loss function does not possess any spurious local minima under $O(n \log^3 n)$ Gaussian random measurements. More specifically, it was shown in \cite{Sun18} that all 
minimizers of $F(\vz)$ coincide with the target signal $\vx$ up to a global phase, and $F(\vz)$ has a negative directional curvature around each saddle point. Thanks to this benign geometric landscape any algorithm which can avoid strict saddle points  converges to the true solution with high probability. 
A trust-region method was employed in \cite{Sun18} to find the global minimizers with random initialization. The results in  \cite{Sun18} require $m\ge O(n\log^3 n)$ samples to guarantee the favorable geometric property and efficient recovery.  On the other hand, based on ample numerical evidences, the authors of \cite{Sun18} conjectured that the optimal sampling complexity could be  $O(n\log n)$ or even $O(n)$ to guarantee the benign landscape of the loss function $F(\vz)$ (cf.  p. 1160 therein).

In this paper, we focus on this conjecture and prove that the loss function $F(\vz)$ possesses the favorable geometric property, as long as the measurement number $m\ge O(n\log n)$, by some sophisticated analysis. In other words, we prove that (1) all local minimizers of the loss function $F(\vz)$ are global; and (2) the objective function $F(\vz)$ has a negative curvature around each saddle point and local maximizer. This is a step toward proving the open problem.

We shall emphasize that if allowing some modifications to the loss function $F(\vz)$, the sampling complexity can be reduced to the optimal bound $O(n)$ \cite{2020b,2020c,cai2019}. In \cite{cai2019} the authors show that a  combination of the loss function (\ref{eq:mod1}) with a judiciously chosen activation function also has the benign geometry structure  under $O(n)$ Gaussian random measurements. Furthermore,  in our recent work \cite{2020a},  we consider another new smoothed amplitude flow estimator which is based
on a piece-wise smooth modification to the loss function
\begin{equation}\label{eq:mod2}
F(\vz)=\sum_{j=1}^m \xkh{\abs{\nj{\va_j,\vz}}-y_j}^2,
\end{equation}
and we could also prove that the loss function \eqref{eq:mod2} after some modifications has a benign geometric landscape under the
optimal sampling threshold $m=O(n)$.

The emerging concept of a benign geometric landscape has also recently been   explored in many other applications of signal processing and machine learning, e.g. matrix sensing \cite{bhojanapalli2016global,park2016non}, tensor decomposition \cite{ge2016matrix}, dictionary learning\cite{sun2016complete} and matrix completion \cite{ge2015escaping}. 
For general optimization problems there exist a plethora of loss functions with 
well-behaved geometric landscapes such that all local optima are also global optima and 
each saddle point has a negative direction curvature in its vincinity. 
Correspondingly several techniques have been developed to guarantee that the standard gradient based optimization algorithms can escape such saddle points efficiently, see e.g. \cite{jin2017escape,du2017gradient,jin2017accelerated}.

\subsection{Our contributions}
In this paper, we focus on the open problem:  {\em what is the optimal sampling complexity to guarantee the loss function $F(\vz)$ given in \eqref{eq:mod1} has favorable geometric landscape?}
We develop several new techniques and prove that $m\ge O(n\log n)$ Gaussian random 
measurements are enough. While we can not prove the optimality of this bound, it is an improvement over the result of $m\ge O(n \log^3 n)$ given in  \cite{Sun18}. 
The main result of our paper is the following theorem.

\begin{theorem} \label{thmA}
Assume that $\va_j \in \R^n, j=1,\ldots,m $ are i.i.d.  standard Gaussian random vectors and $0\ne \vx \in \R^n$ is a fixed vector.
There exist positive absolute constants $C$, $c$ and $c'$, such that if $m\ge C n \log n $, then
with probability at least $1- cm^{-1}-7\exp(-c' m)$ the loss function $F(\vz)$ 
defined by \eqref{eq:mod1} has no spurious local minimizers. In other words, the only local minimizer is $ \vx$ up to a global phase and all saddle points are strict, i.e., each saddle point has a neighborhood
where the function has negative directional curvature. Moreover,  the loss function is strongly convex in a neighborhood of $\pm \vx$, and the point $\vz=0$ is a local maximum point where the Hessian  is strictly
negative-definite.
\end{theorem}
\begin{remark}
 For simplicity we consider here only the real-valued case and will investigate the complex-valued case elsewhere.  In 
Theorem \ref{thmA} the probability bound $O(m^{-1})$ can be refined.
\end{remark}
\begin{remark}
Another interesting issue is to show the measurements are non-adaptive, i.e., a single realization
of measurement vectors $\va_j \in \R^n$ can be used to reconstruct all $0\ne \vx \in \mathbb R^n$. However
we shall not dwell on this refinement here for simplicity.
\end{remark}

\subsection{Notations}
Throughout the paper,  we write $\vz \in \mathbb S^{n-1}$ if $\vz \in \mathbb R^n$ and 
$\|\vz \|_2 =1$.  We use $\chi$ to denote the usual characteristic function. For example  $\chi_A (x)=1$ if $x \in A$ and $\chi_A(x)=0$ if $x\notin A$. 
 For any quantity $X$, we shall write $X=O(Y)$ if $|X| \le C Y$ for some constant
$C>0$.  We write $X\lesssim Y$ if $X\le CY$ for some universal constant $C>0$.  
We shall write $X \ll Y$ if $ X \le c Y$ where the constant $c>0$ will be sufficiently
small.  We use $m\gtrsim n $ to denote $m\ge C n$ where $C>0$ is a universal constant.
In this paper, we use $C, c$ and the subscript (superscript) form of them to denote universal constants whose values vary with the context.

\subsection{Organization}
The rest of the papers are organized as follows.  In Section 2, we divide the whole space $\R^n$ into several regions and investigate the geometric property of $F(\vz)$ on each region.
 In Section 3, we present the detailed justification for the technical lemmas given in Section 2. Finally, the appendix  collects some auxiliary estimates needed in the proof.

\section{Proof of the main result} \label{sec:mainlemma}
In the rest of this section we shall carry out the proof of Theorem \ref{thmA} in several steps.
More specifically, we decompose $\R^n$ into several regions (not necessarily non-overlapping), on each of which $F(\vz)$ has certain property that
 will allow us to show that with high probability $F(\vz)$ has no local minimizers other than $\pm \vx$.  
 Furthermore, we show $F(\vz)$ is strongly convex in a neighborhood of $\pm \vx$.

Without loss of generality, we assume $\norm{\vx}=1$. Denote $\sigma =\sigma(\vz):=\nj{\vz,\vx}/\|\vz\|\|\vx\| $. Then we can decompose $\R^n$ into three regions as shown below.

\begin{itemize}
\item  $\mathcal{R}_1:=\dkh{\vz\in \R^n:  \abs{\sigma} \le  \sqrt{\frac{\sqrt{3}-1}{2}}-\varepsilon_0}$, \vspace{1em}
\item  $\mathcal{R}_2:=\dkh{\vz\in \R^n: \abs{\sigma} \ge 0.5  \quad \mbox{and} \quad \mbox{dist}(\vz,\vx) \ge \delta_0}$, \vspace{1em}
\item  $\mathcal{R}_3:=\dkh{\vz\in \R^n: \mbox{dist}(\vz,\vx) \le \delta_0} $,
\end{itemize}
 where  $\varepsilon_0$ is arbitrary small positive constant and $0< \delta_0< 1/4$ is a universal constant.
 Figure \ref{figure:partition}  visualizes the partitioning regions described above and gives the idea of how they cover the whole space.
 \begin{figure}[H]
\centering
    \subfigure[]{
     \includegraphics[width=0.45\textwidth]{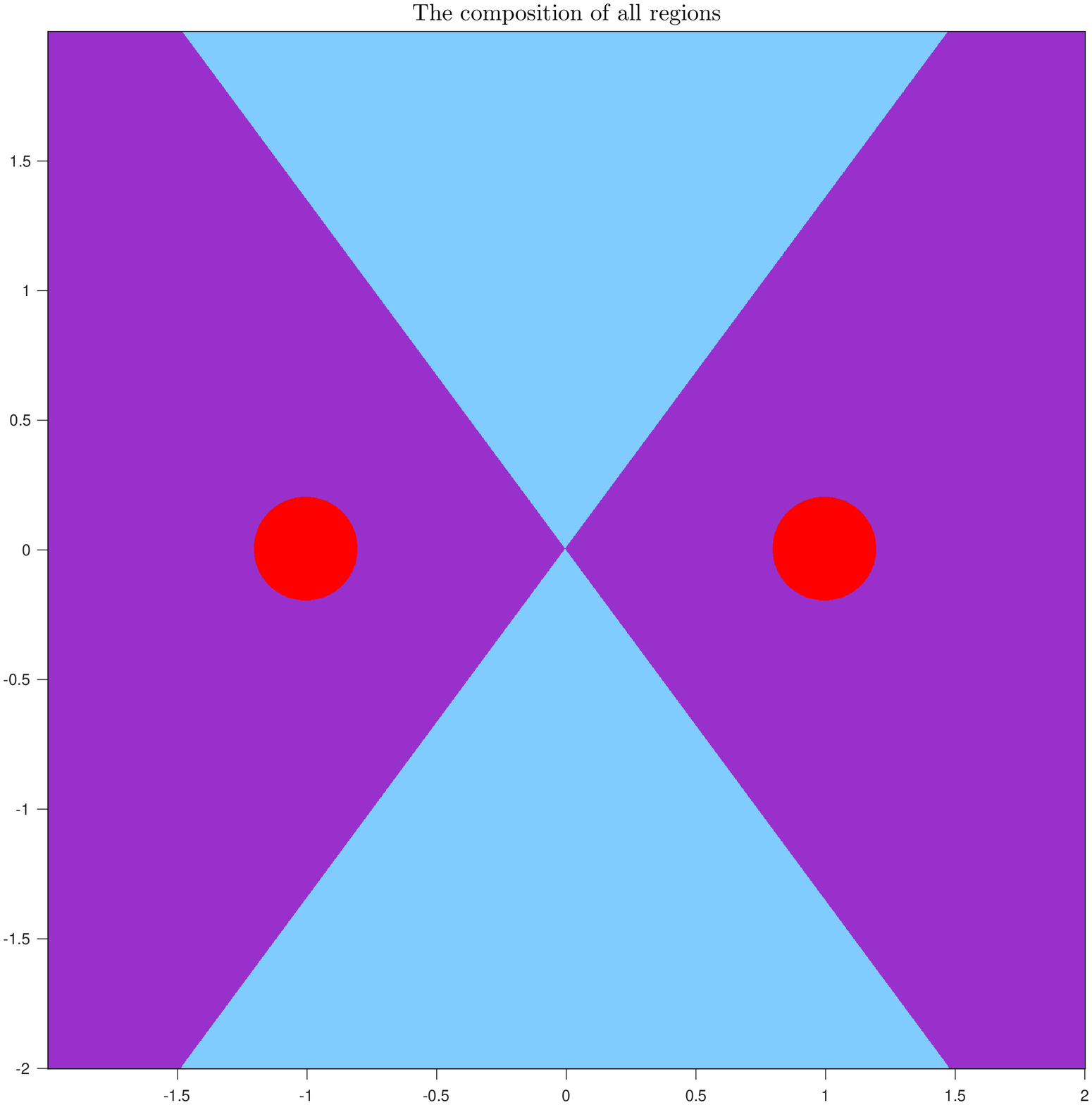}}
\subfigure[]{
     \includegraphics[width=0.45\textwidth]{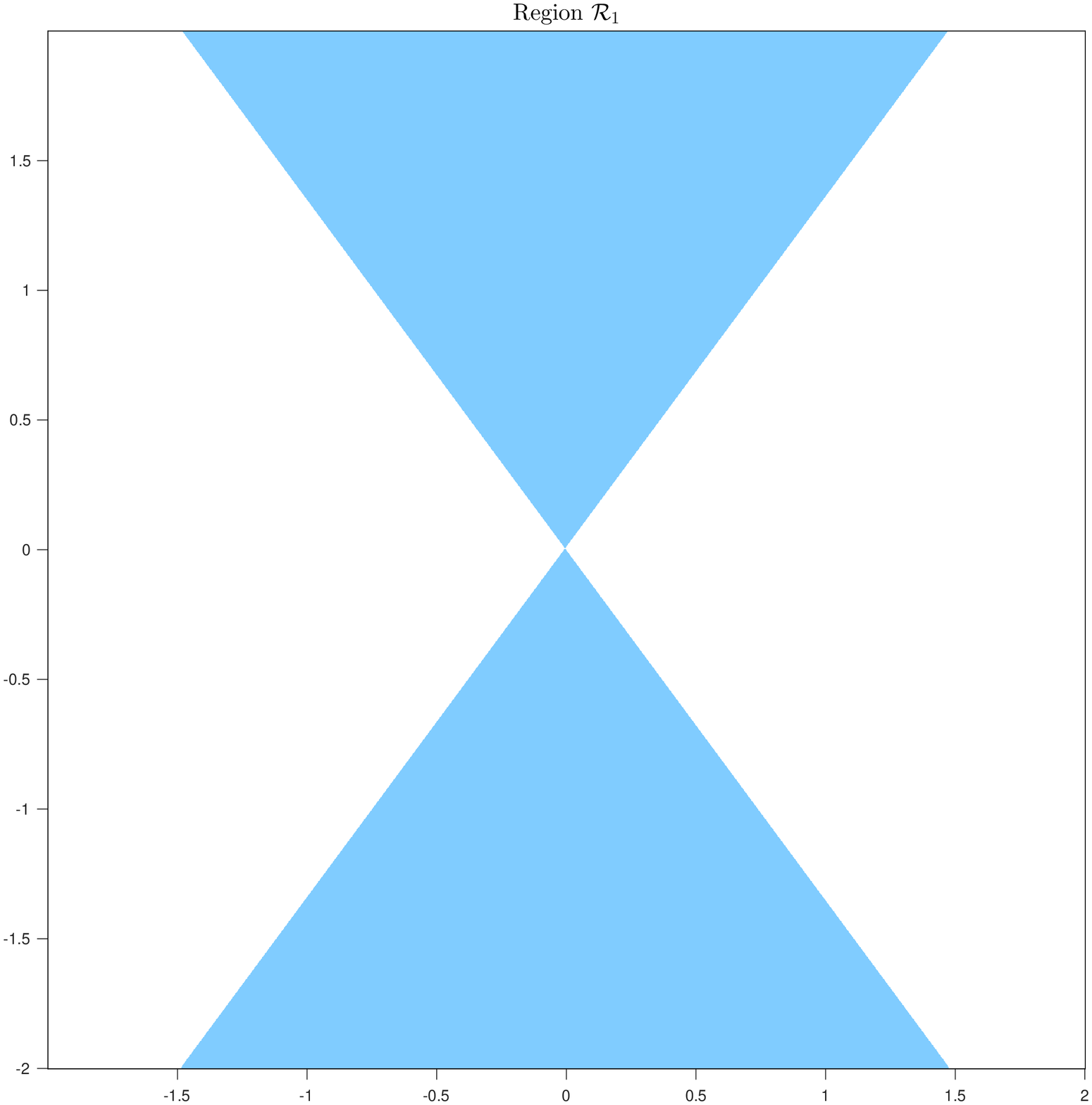}}
\subfigure[]{
     \includegraphics[width=0.45\textwidth]{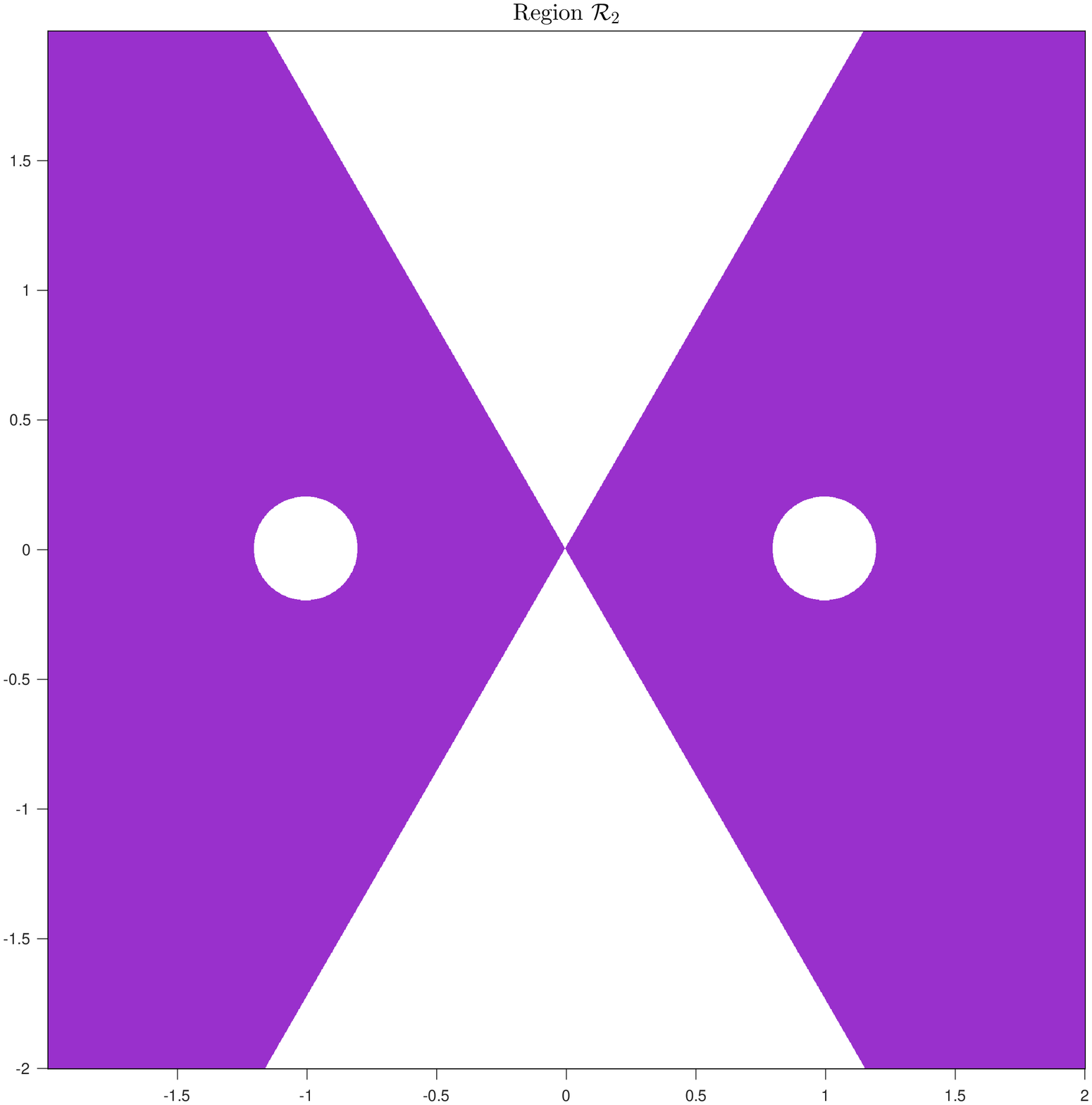}}
     \subfigure[]{
     \includegraphics[width=0.45\textwidth]{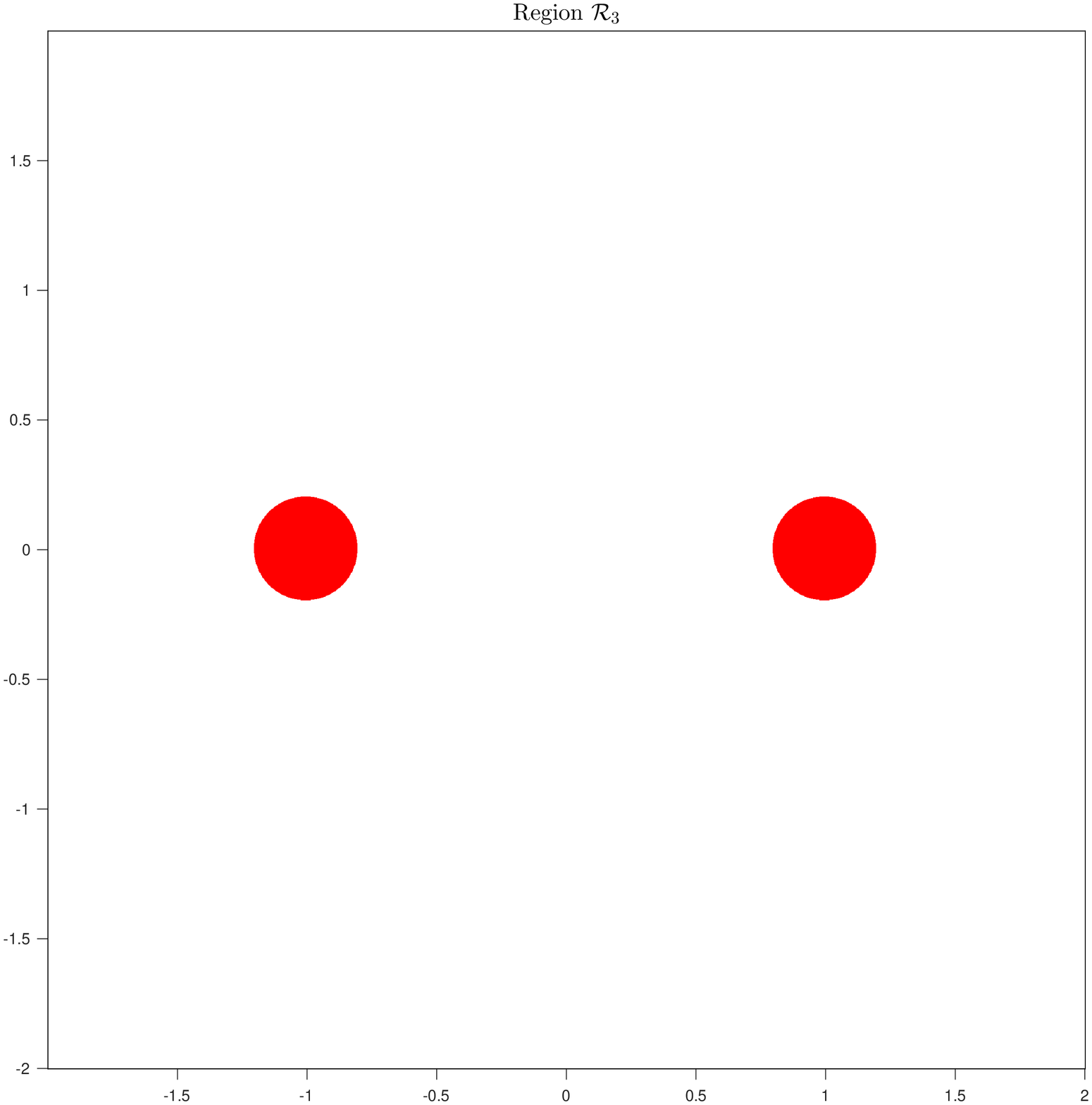}}
\caption{Partition of $\R^2$: The target signal is $\vx=[1,0]$ with constants $\varepsilon_0=0.1$ and $\delta_0=0.2$. (a): The composition of all regions; (b): The region $\mathcal{R}_1$; (c): The region $\mathcal{R}_2$; (d): The region $\mathcal{R}_3$.}
\label{figure:partition}
\end{figure}
The properties of $F(\vz)$ over these regions are summarized in the following three lemmas.

\begin{lemma}  \label{lem:non-vanishing-gradfZ}
 For any $\varepsilon_0>0$ there exists a constant $\delta_1>0$ such that with probability at least $1-c m^{-2}-5\exp(-c(\varepsilon_0) m)$ it holds:  any critical point $\vz \in \mathcal{R}_1:=\dkh{\vz\in \R^n:  \abs{\sigma} \le  \sqrt{\frac{\sqrt{3}-1}{2}}-\varepsilon_0}$ obeys
 \[
 \vx^\T \nabla^2 F(\vz) \vx \le -\delta_1
 \]
provided $m \geq C(\varepsilon_0) n$.  Here, $C(\varepsilon_0)$ and $c(\varepsilon_0)$ are positive constants depending only on $\varepsilon_0$, and $c>0$ is a universal constant.
\end{lemma}

\begin{lemma}  \label{lem:negativeGradfX}
  Assume that $m \geq C(\delta_0) n$.  Then with probability at least $1-c(\delta_0) m^{-2}-\exp(-c'(\delta_0)m)$ there is no critical point in the region $\vz \in \mathcal{R}_2:=\dkh{\vz\in \R^n: \abs{\sigma} \ge 0.5  \quad \mbox{and} \quad \mbox{dist}(\vz,\vx) \ge \delta_0}$.  Here, $C(\delta_0), c(\delta_0)$ and $c'(\delta_0)$ are constants depending only on $\delta_0$. 
\end{lemma}

\begin{lemma}  \label{lem:convex}
 Assume that $m \geq C(\delta_0) n \log n$. There exists a constant $\delta_2$ such that with probability at least  $1-c(\delta_0) m^{-2}-\exp(-c'(\delta_0)m)$ we have $ \vu^\T \nabla^2 F(\vz) \vu \geq \delta_2$ for all $\vz \in  \mathcal{R}_3:=\dkh{\vz\in \R^n: \mbox{dist}(\vz,\vx) \le \delta_0}$ and unit vectors $\vv \in\Rn$. In other words, $F(\vz)$ is strongly convex in $\mathcal{R}_3$.  Here, $C(\delta_0), c(\delta_0)$ and $c'(\delta_0)$ are constants depending only on $\delta_0$. 
 \end{lemma}

The proofs of the above lemmas are given in Section 3.  Lemma \ref{lem:negativeGradfX} guarantee the gradient of $F(\vz)$ does not vanish in $\mathcal{R}_2$. Thus the critical points of $F(\vz)$ can only occur in $\mathcal{R}_1$ and $\mathcal{R}_3$. However, Lemma \ref{lem:non-vanishing-gradfZ} shows that at any critical point in $\mathcal{R}_1$, $F(\vz)$ has a negative directional curvature. Finally, Lemma \ref{lem:convex} implies that  $F(\vz)$ is strongly convex in $\mathcal{R}_3$. Recognizing that  $\nabla F( \vx)=0$ and $ \vx \in \mathcal{R}_3$, thus $ \vx$ is the local minimizer. 
Putting it all together, we can establish Theorem \ref{thmA} as shown below.

{\noindent\it \textbf{Proof of Theorem} \ref{thmA}}~
For any $\vz $ being a possible critical point and satisfying $ |\sigma | \le \sqrt{\frac{\sqrt 3-1} 2} -0.01$, Lemma \ref{lem:non-vanishing-gradfZ} shows that $F(\vz)$ has a negative directional curvature. For any $\vz$ satisfying $ |\sigma | \ge 0.5$ and $\mbox{dist}(\vz,\vx) \ge 0.01$,   Lemma \ref{lem:negativeGradfX} demonstrates that the gradient $\nabla F(\vz)\neq 0$. Finally, when $\vz$ is very close to the target solutions $\pm \vx$, $F(\vz)$ is strongly convex and $\pm \vx$ are the global solutions.
\qed

\section{Proofs of  technical results in Section \ref{sec:mainlemma}}
The basic idea of the proof is to show for each critical point except $\pm \vx$ there is a negative curvature direction. 

\subsection{Proof of Lemma \ref{lem:non-vanishing-gradfZ}}
{\noindent\it \textbf{Proof}}~
For any $\vz \ne 0$, denote 
\begin{align} \label{eq:nota}
&\vz =\sqrt R \hat \vz \quad \mbox{with} \quad
\hat \vz \in \mathbb S^{n-1}.
\end{align}
Recall the function 
\begin{equation} \label{eq:loss1}
F(\vz)=\frac 1m \sum_{j=1}^m \xkh{\abs{\nj{\va_j,\vz}}^2-\abs{\nj{\va_j,\vx}}^2}^2.
\end{equation}
Through a simple calculation,   the Hessian  of the function $F(\vz)$ along the $\xi$ direction can be denote by
\begin{align*}
H_{\xi\xi}(\vz):= \xi^\T \nabla^2 F(\vz) \xi  & = \frac 4m \sum_{j=1}^m (\va_j^\T \xi)^2 \xkh{3(\va_j^\T \vz)^2 - (\va_j^\T \vx)^2}.
\end{align*}
We first show that $\vz=0$ is a local maximum.  Indeed, 
by Corollary \ref{bp9.01} in the appendix,
if $m\gtrsim n$ then it holds with probability at least $1-\exp(-cm)$ that 
\begin{align*}
H_{\xi\xi}(0) = - \frac 4m \sum_{j=1}^m (\va_j^\T \xi)^2 (\va_j^\T \vx)^2  \le -c_1<0, \qquad\forall\, \xi \in \mathbb S^{n-1}.
\end{align*}
Here, $c$ and $c_1$ are universal positive constants. This means that with high probability the Hessian $\nabla^2 f(0)$
is strictly negative definite.  

Next, we consider the case where $\vz \neq 0$ and prove the loss function \eqref{eq:loss1} has a negative curvature at each critical point in the regime
$\mathcal{R}_1=|\nj{\hat \vz, \vx} | \le  \sqrt{\frac{\sqrt 3-1} 2} -\varepsilon_0 $. Through a simple calculation, we have
\[
\nabla F(\vz)=\frac 2m \sum_{j=1}^m \xkh{\abs{\nj{\va_j,\vz}}^2-\abs{\nj{\va_j,\vx}}^2} \va_j \va_j^\T \vz.
\]
If at some $\vz \neq 0$ we have a critical point, then 
\[
\nj{\nabla F(\vz), \vz}=\frac {2R}m \sum_{j=1}^m (\va_j^\T \hz)^4 - \frac 2m  \sum_{j=1}^m (\va_j^\T \hz)^2 (\va_j^\T \vx)^2=0,
\]
where $R$ is defined by \eqref{eq:nota}.  By Lemma \ref{bp8.02}, if $m\gtrsim n$ then it holds with probability at least $1-\exp(-c_2 m)$ that $\frac {1}m \sum_{j=1}^m (\va_j^\T \hz)^4 \ge 1$. Here, $c_2>0$ is a universal constant. 
 Consequently, if at some $\vz \ne 0$ we have a critical point, then it holds
\begin{equation} \label{Au26e0-1}
R=\frac{\frac 1 m  \sum_{j=1}^m (\va_j^\T \hz)^2 (\va_j^\T \vx)^2}{\frac {1}m \sum_{j=1}^m (\va_j^\T \hz)^4}.
\end{equation}
On the other hand, the Hessian at this point along the direction $\vx$ is
\[
\frac 14 H_{\vx \vx}(\vz)= 3 R \cdot \frac 1m \sum_{j=1}^m  (\va_j^\T \hat \vz)^2 (\va_j^\T \vx)^2
- \frac 1 m \sum_{j=1}^m (\va_j^\T \vx)^4.
\]
Using the equation \eqref{Au26e0-1}, we obtain
\begin{equation} \label{eq:hee1}
   \frac 14 H_{\vx \vx}(\vz)\cdot \frac {1}m \sum_{j=1}^m (\va_j^\T \hz)^4 =3\Big(\frac 1 m  \sum_{j=1}^m (\va_j^\T \hz)^2 (\va_j^\T \vx)^2\Big)^2-\frac {1}m \sum_{j=1}^m (\va_j^\T \hz)^4  \cdot \frac 1 m \sum_{j=1}^m (\va_j^\T \vx)^4.
\end{equation}
We claim that  for any $0<\epsilon<1$, when $m\ge C (\epsilon)  n$,  with probability at least $1-\frac{c}{m^2}-\exp(-c(\epsilon) m)$,  the following holds:
\begin{equation} \label{eq:claimB}
\frac 1 m  \sum_{j=1}^m (\va_j^\T \hz)^2 (\va_j^\T \vx)^2 \le 2\sigma^2+1+2\epsilon+2\sqrt{\epsilon}\sqrt{\frac {1}m \sum_{j=1}^m (\va_j^\T \hz)^4},
\end{equation}
where $C(\epsilon), c(\epsilon)>0$ are constants depending only on $\epsilon$ and $c>0$ is a universal constant.
On the other hand, by Lemma \ref{bp8.02}, when $m\ge C(\epsilon) n$,  with probability at least $1-3\exp(-c(\epsilon)m)$ it holds
\begin{equation} \label{eq:lowax4}
\frac {1}m \sum_{j=1}^m (\va_j^\T \vx)^4 \ge 3-\epsilon.
\end{equation}
Putting \eqref{eq:claimB} and \eqref{eq:lowax4} into \eqref{eq:hee1}, we obtain that for $m\ge C(\epsilon) n$,  with probability at least $1-\frac c {m^2}-4\exp(-c(\epsilon)m)$, it holds
\[
\frac 14 H_{\vx \vx}(\vz) \cdot \frac {1}m \sum_{j=1}^m (\va_j^\T \hz)^4 \le 3(2\sigma^2+1)^2 -  \frac {3-13\epsilon}m \sum_{j=1}^m (\va_j^\T \hz)^4 +60\sqrt\epsilon \sqrt{\frac {1}m \sum_{j=1}^m (\va_j^\T \hz)^4}+24\epsilon.
\]
Since the term $\frac {1}m \sum_{j=1}^m (\va_j^\T \hz)^4$ is the sum of nonnegative random variables, the deviation below its expectation is bounded and the lower-tail
is well-behaved. More concretely, by Lemma \ref{bp8.02}, if $m\ge C(\epsilon) n$ then  with probability at least $1-3\exp(-c(\epsilon) m)$ we have 
\begin{align*}
\frac {1}m \sum_{j=1}^m (\va_j^\T \hz)^4 \ge 3-\epsilon >1 , \qquad\forall\, \hat \vz  \in \mathbb S^{n-1}.
\end{align*}
It immediately gives 
\begin{eqnarray*}
\frac 14 H_{\vx \vx}(\vz) \cdot \frac {1}m \sum_{j=1}^m (\va_j^\T \hz)^4 &\le& 3(2\sigma^2+1)^2 -(3-13\epsilon-60\sqrt\epsilon)\cdot \frac {1}m \sum_{j=1}^m (\va_j^\T \hz)^4 +24\epsilon\\
&\le &3(2\sigma^2+1)^2-9 + 66\epsilon+180\sqrt\epsilon \\
&\le & - c_0 \varepsilon_0 + 66\epsilon+180\sqrt\epsilon \\
&< & -\delta_1
\end{eqnarray*}
for some constant $\delta_1>0$
by taking $\epsilon$ to be sufficiently small (depending on $\varepsilon_0$), provided $\abs{\sigma}\le \sqrt{\frac{\sqrt 3-1} 2}-\varepsilon_0$. Here,  $c_0$ is an absolute constant.
This means the Hessian matrix has a negative curvature along the direction $\vx$,  which  proves the lemma.

Finally, it remains to prove the claim \eqref{eq:claimB}.  Due to the heavy tail of fourth powers of Gaussian random variables, to prove the result with sampling complexity $m\gtrsim n$, we need to decompose it into several parts by a Lipschitz continuous truncated function. To do this,   take $\phi \in C_c^{\infty}(\mathbb R)$,  $0\le \phi\le 1$ for all
$z\in \mathbb R$,  $\phi\equiv 1$ for $|z|\le 1$
and $\phi \equiv 0$ for $|z|>2$. We can write
\begin{eqnarray*}
 \frac 1 m  \sum_{j=1}^m (\va_j^\T \hz)^2 (\va_j^\T \vx)^2   &=&  \frac 1m \sum_{j=1}^m (\va_j^\T \vx)^2 
\phi(\frac{\va_j^\T \vx} N)
(\va_j^\T \hz)^2 \phi(\frac{\va_j^\T \hz} N)  \\
&& + \frac 1m \sum_{j=1}^m (\va_j^\T \vx)^2 
(1-\phi(\frac{\va_j^\T \vx} N))
(\va_j^\T \hz)^2 \phi(\frac{\va_j^\T \hz} N)  \\
&&  +\frac 1m \sum_{j=1}^m (\va_j^\T \vx)^2 
(\va_j^\T \hz)^2 (1-\phi(\frac{\va_j^\T \hz} N) )\\
& :=& B_1 +B_2+B_3.
\end{eqnarray*}

Next, we give upper bounds for the terms $B_1, B_2$ and $B_3$. Thanks to the smooth cut-off, $B_1$ can be well bounded. By Lemma \ref{bp9},   
for any $0<\epsilon<1/2$,  there exist constants $C',c'>0$ depending on $N$ such that if  $m\ge C' \epsilon^{-2} \log(1/\epsilon) n$ then with probability at least $1-\exp(-c' \epsilon^2 m)$ it holds that 
\begin{align*}
\Abs{\frac 1m \sum_{j=1}^m (\va_j^\T \vx)^2 
\phi(\frac{\va_j^\T \vx} N)
(\va_j^\T \hz)^2 \phi(\frac{\va_j^\T \hz} N)  -\mathbb E  (\va_1^\T \vx)^2 
\phi(\frac{\va_1^\T \vx} N)
(\va_1^\T \hz)^2 \phi(\frac{\va_1^\T \hz} N) } \le \epsilon, \qquad \forall\, \hat \vz \in \mathbb S^{n-1}.
\end{align*}
Moreover,  note that 
\begin{eqnarray*}
&& \Abs{ \E (\va_1^\T \vx)^2 
\phi(\frac{\va_1^\T \vx} N)
(\va_1^\T \hz)^2 \phi(\frac{\va_1^\T \hz} N)- \E (\va_1^\T \vx)^2 (\va_1^\T \hz)^2 } \\
&=& \Abs{\mathbb E  (\va_1^\T \vx)^2 (\va_1^\T \hz)^2
\cdot \xkh{ ( \phi(\frac {\va_1^\T \vx} N) -1) \phi(\frac {\va_1^\T \hz} N) + 
\phi(\frac{\va_1^\T \hz} N) -1} } \\
&\le & \mathbb E (\va_1^\T \vx)^2 (\va_1^\T \hz)^2
\cdot (\chi_{|\va_1^\T \vx| \ge N} + \chi_{|\va_1^\T \hz|\ge N} ).
\end{eqnarray*}
Since $\E (\va_1^\T \vx)^2 (\va_1^\T \hz)^2 = 2 \sigma^2 +1$, it then follows from Lemma \ref{bp4} that  for $N$ sufficiently large (depending only
on $\epsilon$) we have
\[
\E (\va_1^\T \vx)^2 
\phi(\frac{\va_1^\T \vx} N)
(\va_1^\T \hz)^2 \phi(\frac{\va_1^\T \hz} N) \le \E (\va_1^\T \vx)^2 (\va_1^\T \hz)^2+\epsilon=2 \sigma^2 +1+\epsilon,
\]
which means 
\[
B_1:=\frac 1m \sum_{j=1}^m (\va_j^\T \vx)^2 
\phi(\frac{\va_j^\T \vx} N)
(\va_j^\T \hz)^2 \phi(\frac{\va_j^\T \hz} N)  \le 2\sigma^2+1+2\epsilon.
\]
For the terms $B_2$ and $B_3$,  when $N$ is sufficiently large depending only
on $\epsilon$, applying Lemma \ref{bp8} gives 
 \begin{align*}
&\frac 1m \sum_{j=1}^m (\va_j^\T \vx )^4 \chi_{|\va_j^\T \hz |>N} \le  \epsilon,
\qquad \forall\, \hat \vz \in \mathbb S^{n-1}
\end{align*}
with probability at least $1-\frac{c''}{m^2}-\exp(-c'''\epsilon^2 m)$ provided $m\ge C \epsilon^{-4} \log(1/\epsilon) n$. Here, $C,c''$ and $c'''$ are universal positive constants.
Thus for $B_2$ and $B_3$,  by Cauchy-Schwarz inequality,  we have 
\begin{align*}
B_2 &\le \;  \sqrt{\frac 1 m \sum_{j=1}^m (\va_j^\T \hz)^4}
\sqrt{ \frac 1m \sum_{j=1}^m (\va_j^\T \vx )^4 \chi_{|\va_j^\T \vx|>N} }
\le  \sqrt \epsilon \sqrt{\frac 1 m \sum_{j=1}^m (\va_j^\T \hz)^4} ; \\
B_3 &\le \; \sqrt{\frac 1 m \sum_{j=1}^m (\va_j^\T \hz)^4}
( \frac 1m \sum_{j=1}^m (\va_j^\T \vx)^4 \chi_{|\va_j^\T \hz|>N} )
\le  \sqrt\epsilon  \sqrt{\frac 1 m \sum_{j=1}^m (\va_j^\T \hz)^4} .
\end{align*}
Collecting the above estimators together gives that when $m\ge C (\epsilon)  n$,  with probability at least $1-\frac{c}{m^2}-\exp(-c(\epsilon) m)$, it holds
\[
\frac 1 m  \sum_{j=1}^m (\va_j^\T \hz)^2 (\va_j^\T \vx)^2 \le 2\sigma^2+1+2\epsilon+2\sqrt{\epsilon}\sqrt{\frac {1}m \sum_{j=1}^m (\va_j^\T \hz)^4},
\]
which completes the proof of claim \eqref{eq:claimB}.
\qed

\subsection{Proof of Lemma \ref{lem:negativeGradfX}}
{\noindent\it \textbf{Proof}}~
Without loss of generality we can assume $\sigma:=\nj{\hz,\vx}\ge 0$.
For any $\vz \ne 0$, denote 
\begin{align} \label{eq:nota1}
&\vz =\sqrt R \hat \vz \quad \mbox{with} \quad
\hat \vz \in \mathbb S^{n-1}.
\end{align}
Recognize that 
\[
\nabla F(\vz)=\frac 2m \sum_{j=1}^m \xkh{\abs{\nj{\va_j,\vz}}^2-\abs{\nj{\va_j,\vx}}^2} \va_j \va_j^\T \vz.
\]
At any potential critical point $\vz = \sqrt R \hat \vz$, we should have $\nabla F(\vz)  =0$. Thus $\nj{\nabla F(\vz),\vz}=0$ gives
\[
 \frac Rm \sum_{j=1}^m (\va_j^\T \hz )^4 =
  \frac 1 m \sum_{j=1}^m (\va_j^\T \hz)^2 (\va_j^\T \vx)^2.
\]
Similarly, according to $\nj{\nabla F(\vz),\vz}=0$ we have
\[
\frac R m \sum_{j=1}^m (\va_j^\T \hz)^3 (\va_j^\T \vx)
= \frac 1 m \sum_{j=1}^m (\va_j^\T \hz)  (\va_j^\T \vx)^3 
\]
Combining the above two equations leads to the following fundamental relation for any critical point $\vz=\sqrt R \hat \vz$:
\begin{align} \label{c_fun1}
 \frac 1 m \sum_{j=1}^m (\va_j^\T \hz)^2 (\va_j^\T \vx)^2 
\cdot  \frac 1m \sum_{j=1}^m (\va_j^\T \vx) (\va_j^\T \hz)^3 
 =  \frac 1m \sum_{j=1}^m (\va_j^\T \hz )^4 
\cdot   \frac 1 m \sum_{j=1}^m (\va_j^\T \hz)  (\va_j^\T \vx)^3 .
\end{align}
Observe that 
\[
 \mathbb E \frac 1 m \sum_{j=1}^m (\va_j^\T \hz)  (\va_j^\T \vx)^3= 3\sigma^2,
\]
where $\sigma:=\nj{\hz,\vx}\ge 0$. By Corollary \ref{bp11.01}, for any $0<\epsilon<1/2$, if $m\ge C(\epsilon) n$ then with probability at least $1-\frac{c(\epsilon)}{m^2}$ it holds
\begin{equation} \label{eq:F13}
\Abs{ \frac 1 m \sum_{j=1}^m (\va_j^\T \hz)  (\va_j^\T \vx)^3 -  3\sigma}\le \epsilon.
\end{equation}
For the convenience, we denote 
\[
A:=\frac 1m \sum_{j=1}^m (\va_j^\T \hz )^4,\qquad B:= \frac 1 m \sum_{j=1}^m (\va_j^\T \hz)^2 (\va_j^\T \vx)^2 \quad \mbox{and} \quad A_1:=\frac 1m \sum_{j=1}^m (\va_j^\T \vx) (\va_j^\T \hz)^3 .
\]
We claim that for any $0<\epsilon<1/2$, if $m\ge C(\epsilon) n$ then with probability at least $1-\frac{c(\epsilon)}{m^2}-\exp(-c'(\epsilon)m)$ the following holds:
\begin{equation} \label{eq:claB}
 \abs{B -  2\sigma^2-1} \le 2\epsilon+2 \epsilon^{\frac23} A^{\frac13}
\end{equation}
and 
\begin{equation} \label{eq:clamA1}
 \abs{A_1 -  3\sigma} \le 2\epsilon+ \epsilon^{\frac13} A^{\frac23}.
\end{equation}
Putting \eqref{eq:F13}, \eqref{eq:claB} and \eqref{eq:clamA1} into \eqref{c_fun1}, we immediately have 
\begin{align} \label{Au26_tmp002}
(2\sigma^2+1 +O(\epsilon) A^{\frac 13} ) \cdot(3\sigma+O(\epsilon) A^{\frac 23} )=
 A (3\sigma+ O(\epsilon)).
\end{align}
with probability at least $1-\frac{c(\epsilon)}{m^2}-\exp(-c'(\epsilon)m)$, provided $m\ge C(\epsilon) n$. 
By Lemma \ref{bp8.02}, for $m\ge C(\epsilon) n$ with probability at least $1-\exp(-c'(\epsilon)m )$ it holds that
\begin{align} \label{boundA}
A\ge 3-\epsilon.
\end{align}
In particular, we have $A\ge 1$. Then we can simplify
\eqref{Au26_tmp002} as
\begin{align} \label{Au26_tmp003}
3\sigma(2\sigma^2+1) = A \cdot (3 \sigma +O(\epsilon) )
\ge (3-\epsilon) (3\sigma +O(\epsilon) ).
\end{align}
Note that $\vz \in \mathcal{R}_2$, which means $\sigma \ge 0.5$. On the other hand, recall that $\sigma \le 1$. By taking $\epsilon>0$ to be sufficiently small, it then follows from
\eqref{Au26_tmp003} that $\sigma$ must be sufficiently close to $1$. It implies that for any $0<\eta<1$, if $m\ge C(\eta) n$ then with probability at least $1-\frac{c(\eta)}{m^2}-\exp(-c'(\eta)m )$ it holds
\begin{equation} \label{eq:sigmbound}
\sqrt{1-\eta} \le \sigma \le \sqrt{1-\eta}.
\end{equation}
Furthermore,  it follows from the equality in \eqref{Au26_tmp003} that
\begin{align*}
A \le  2\sigma^2+1 + \frac 1{\sigma} O(\epsilon) \le 3 +\eta. 
\end{align*}
Combining with \eqref{boundA} gives the  desired two-way bound for $A$ that 
\[
3-\eta \le A \le 3+\eta.
\]
On the other hand, it follows from \eqref{eq:claB}  that if $m\ge C(\epsilon) n$ then with probability at least $1-\frac{c(\epsilon)}{m^2}-\exp(-c'(\epsilon)m)$ it holds
\begin{align*}
B= 2\sigma^2+1 +O(\epsilon)A^{\frac 13}.
\end{align*}
This immediately means that the term $B$ also has the desired two-way bounds 
\[
3-\eta \le B \le 3+\eta.
\]
Finally if $0\ne \vz=\sqrt R \hat \vz $ is a critical point, then by \eqref{Au26e0-1}, we have
\begin{align*}
R= \frac B A.
\end{align*}
Since we have already shown that $B$ and $A$ are well-bounded, it then follows that 
\begin{align} \label{eq:Rbound}
1-\eta \le R \le 1+\eta.
\end{align}
Combining \eqref{eq:sigmbound} and \eqref{eq:Rbound}, we obtain that if $\vz:=\sqrt R \hz$ is a critical point then it holds
\[
\mbox{dist}(\vz,\vx)=\sqrt{R+1-2\sqrt{R}\sigma} \le \sqrt{3\eta} \le \delta_0
\]
by taking $\eta:=\delta_0^2/3$. This contradicts to the condition that $\mbox{dist}(\vz,\vx) \ge \delta_0$ for all $\vz \in \mathcal{R}_2$. 
Thus, the loss function $F(\vz)$ has no critical point on $ \mathcal{R}_2$. We arrive at the conclusion.

Finally, it remains to prove the claims \eqref{eq:claB} and \eqref{eq:clamA1}.
Let $\phi\in C_c^{\infty}(\mathbb R)$ be such that $0\le \phi(x)\le 1$ for all
$x\in \mathbb R$, $\phi(x)=1$ for $|x|\le 1$ and $\phi(x)=0$ for $|x|\ge 2$. 
Then we can write
\begin{align*}
& B=  \frac 1m \sum_{j=1}^m (\va_j^\T \hz)^2 (\va_j^\T \vx)^2  
\phi(\frac {\va_j^\T \hz} N) \phi(\frac{\va_j^\T \vx} N) + r_B, \\
& A_1=\frac 1m \sum_{j=1}^m (\va_j^\T \hz)^3 (\va_j^\T \vx)
\phi(\frac {\va_j^\T \hz} N) +r_1,
\end{align*}
where
\begin{align*}
& |r_B| \le \frac 1m \sum_{j=1}^m (\va_j^\T \hz)^2 (\va_j^\T \vx)^2  
(\chi_{|\va_j^\T \hz|\ge N} + \chi_{|\va_j^\T \vx|\ge N} ) ; \\
& |r_1| \le \frac 1m \sum_{j=1}^m |\va_j^\T \hz |^3 |\va_j^\T \vx|
\cdot \chi_{|\va_j^\T \hz | \ge N}.
\end{align*}
Through a simple calculation, we have
\[
\mathbb E  \frac 1 m \sum_{j=1}^m (\va_j^\T \hz)^2 (\va_j^\T \vx)^2 = 
2\sigma^2+1 \quad \mbox{and}\quad \mathbb E \frac 1m \sum_{j=1}^m (\va_j^\T \vx) (\va_j^\T \hz)^3
=3\sigma.
\]
Using the same procedure as the claim \eqref{eq:claimB}, it is easy to derive from Lemma \ref{bp9}, Lemma \ref{bp4} and Lemma \ref{bp12} that 
for any $0<\epsilon<1$ if $m\ge C(\epsilon)n$ then with probability at least $1-\frac{c(\epsilon)}{m^2}-\exp(-c'(\epsilon)m)$ it holds
\[
\Abs{\frac 1m \sum_{j=1}^m (\va_j^\T \hz)^2 (\va_j^\T \vx)^2  
\phi(\frac {\va_j^\T \hz} N) \phi(\frac{\va_j^\T \vx} N)-2\sigma^2-1} \le 2\epsilon
\]
and 
\[
\Abs{\frac 1m \sum_{j=1}^m (\va_j^\T \hz)^3 (\va_j^\T \vx)
\phi(\frac {\va_j^\T \hz} N)-3\sigma} \le 2\epsilon.
\]
To deal with the error terms,  observe that 
\begin{align*}
\sum_{j=1}^m (\va_j^\T \hz)^2 \chi_{|\va_j^\T \hz|>N}
(\va_j^\T \vx)^2 
& =\sum_{j=1}^m |\va_j^\T \hz |^{\frac 23} \chi_{|\va_j^\T \hz|>N}
\cdot (\va_j^\T \vx)^2 \cdot |\va_j^\T \hz|^{\frac 43} \notag \\
&\le ( \sum_{j=1}^m |\va_j^\T \hz| 
\chi_{|\va_j^\T \hz|>N}
|\va_j^\T \vx|^3)^{\frac 23}
( \sum_{j=1}^m |\va_j^\T \hz|^4)^{\frac 13} ; \\
 \sum_{j=1}^m |\va_j^\T \hz|^3 \chi_{|\va_j^\T \hz|>N}
|\va_j^\T \vx|  
& = \sum_{j=1}^m |\va_j^\T \hz|^{\frac 13} \chi_{|\va_j^\T \hz|>N}
|\va_j^\T \vx| \cdot |\va_j^\T \hz|^{\frac 83} \notag \\
&\le ( \sum_{j=1}^m |\va_j^\T \hz| 
\chi_{|\va_j^\T \hz|>N}
|\va_j^\T \vx|^3)^{\frac 13}
( \sum_{j=1}^m |\va_j^\T \hz|^4)^{\frac 23} ; \\
\sum_{j=1}^m (\va_j^\T \hz)^2 \chi_{|\va_j^\T \vx |>N}
(\va_j^\T \vx)^2 
&=\sum_{j=1}^m |\va_j^\T \hz|^{\frac 23} \chi_{|\va_j^\T \vx|>N}
(\va_j^\T \vx)^2 \cdot |\va_j^\T \hz|^{\frac 43} \notag \\
&\le ( \sum_{j=1}^m |\va_j^\T \hz| 
\chi_{|\va_j^\T \vx|>N}
|\va_j^\T \vx|^3)^{\frac 23}
( \sum_{j=1}^m |\va_j^\T \hz|^4)^{\frac 13}.
\end{align*}
Using the  Lemma \ref{bp12} again, we obtain that when $m\ge C(\epsilon)n$, with probability at least $1-\frac{c(\epsilon)}{m^2}-\exp(-c'(\epsilon)m)$ it holds
\[
|r_B| \le 2\epsilon^{\frac23} A^{\frac13} \quad \mbox{and} \quad |r_1| \le \epsilon^{\frac13} A^{\frac23}.
\]
Thus,  we complete the proofs of claim \eqref{eq:claB} and \eqref{eq:clamA1}.
\qed

\subsection{Proof of Lemma \ref{lem:convex}}
This section goes in the direction of showing the loss function is strongly convex in a neighborhood of $\pm \vx$, as demonstrated in Lemma \ref{lem:convex}.

\vspace{2em}
{\noindent\it \textbf{Proof}}~
Recall that along any direction $\vu \in \mathbb S^{n-1}$,
\begin{align*}
\vu^\T \nabla^2 F(\vz) \vu &= \frac 4m \sum_{j=1}^m (\va_j^\T  \vu)^2 \xkh{ 3(\va_j^\T \vz)^2 - (\va_j^\T \vx)^2}.
\end{align*}
To prove the lemma, it suffices to give a lower bound for the first term and an upper bound for the second term. Indeed, for the second term
$\frac 1m \sum_{j=1}^m (\va_j^\T  \vu)^2(\va_j^\T \vx)^2 $, by Lemma \ref{bp14}, for any $0<\epsilon< 1/2$ if $m\ge C(\epsilon) n \log n$ then with probability at least
$1-\frac{c(\epsilon)}{m^2}$ it holds
\[
\frac 1m \sum_{j=1}^m (\va_j^\T  \vu)^2(\va_j^\T \vx)^2 \le  1+2(\vu^\T \vx)^2 +\epsilon \quad \mbox{for any} \quad \vu\in \mathbb S^{n-1}.
\]
Here, we use the fact that $ \mathbb E (\va_1^\T \vu)^2 (\va_1^\T \vx)^2 =1+2 (\vu^\T \vx)^2$.
For the first term $\frac 1m \sum_{j=1}^m (\va_j^\T  \vu)^2(\va_j^\T \vz)^2 $, denote $\vz:=\sqrt R \hz$ where $\hz\in \mathbb S^{n-1}$.
Take $\phi \in C_c^{\infty}(\mathbb R)$ such that
$0\le \phi(x)\le 1$ for all $x$, $\phi(x)=1$ for $|x|\le 1$ and $\phi(x)=0$ for $|x|\ge 2$. It is easy to see
\[
\frac 1m \sum_{j=1}^m (\va_j^\T  \vu)^2(\va_j^\T \vz)^2 \ge \frac R m\sum_{j=1}^m (\va_j^\T \vu)^2 \phi(\frac {\va_j^\T \vu} N)
(\va_j^\T \hz)^2 \phi(\frac {\va_j^\T \hz} N).
\]
By Lemma \ref{bp9},  for any $0<\epsilon<1/2$ when $m\ge C(\epsilon) n$ with probability at least $1-\exp(-c'(\epsilon)m)$, it holds that
\begin{align*}
\Abs{ \frac 1 m\sum_{j=1}^m (\va_j^\T \vu)^2 \phi(\frac {\va_j^\T \vu} N)
(\va_j^\T \hz)^2 \phi(\frac {\va_j^\T \hz} N) -
\mathbb E\xkh{ (\va_1^\T \vu)^2 \phi(\frac {\va_1^\T \vu} N)
(\va_1^\T \hz)^2 \phi(\frac {\va_1^\T \hz} N} } \le \epsilon,
\quad \forall\, \vu, \hz \in \mathbb S^{n-1}.
\end{align*}
On the other hand, it follows from Lemma \ref{bp4} that there exists $N>0$ sufficiently large (depending
only on $\epsilon$) such that
\begin{eqnarray*}
&&\Bigl|\mathbb E( (\va_1^\T \vu)^2 \phi(\frac {\va_1^\T \vu} N)(\va_1^\T \hz)^2 \phi(\frac {\va_1^\T \hz} N) 
- \mathbb E (\va_1^\T \vu)^2 (\va_1^\T \hz)^2 \Bigr| \\
&=& \Abs{\mathbb E  (\va_1^\T \vu)^2 (\va_1^\T \hz)^2
\cdot \xkh{\Big( \phi(\frac {\va_1^\T \vu} N) -1\Big) \phi(\frac {\va_1^\T \hz} N) + 
\phi(\frac{\va_1^\T \hz} N) -1} } \notag\\
&\le & \mathbb E (\va_1^\T \vu )^2 (\va_1^\T \hz)^2
\cdot (\chi_{|\va_1^\T \vu | \ge N} + \chi_{|\va_1^\T \hz|\ge N} )\\
&\le & \epsilon.
\end{eqnarray*}
Collecting the above estimators, we have that  when $m\ge C(\epsilon) n \log n$,  with probability at least
$1-\frac{c(\epsilon)}{m^2}-\exp(-c'(\epsilon ) m)$, it holds
\begin{eqnarray}
\frac 1 4 \vu^\T \nabla^2 F(\vz) \vu & \ge &  3 R \xkh{  1+2(\vu^\T \hz)^2
-2 \epsilon}  - 1- 2(\vu^\T \vx)^2 -\epsilon \nonumber\\
& = & (\vu^\T \vz)^2 -(\vu^\T \vx)^2 + 3R-1 -(6R+1)\epsilon \nonumber\\
&\ge & -\norm{\vz+\vx} \norm{\vz-\vx}+3R-1 -(6R+1)\epsilon \label{eq:nabla2}
\end{eqnarray}
for all $\vu,\hz \in \mathbb{S}^{n-1}$.
Here, we use the fact that $E(\va_1^\T \vu)^2 (\va_1^\T \hz)^2=1+2(\vu^\T \hz)^2$ in the first inequality.
Recall that $\vz \in \mathcal{R}_3$. It means 
\begin{equation} \label{eq:distance}
\mbox{dist}(\vz,\vx) \le \delta_0.
\end{equation}
Without loss of generality we assume $\sigma:=\nj{\hz,\vx}\ge 0$.  It then follows from \eqref{eq:distance} that 
\[
\abs{\sqrt R-1}\le \norm{\vz-\vx}\le \delta_0\quad \mbox{and} \quad \norm{\vz+\vx}  \le 2+\delta_0.
\]
Putting it into \eqref{eq:nabla2}, we have 
\[
\frac 1 4 \vu^\T \nabla^2 F(\vz) \vu \ge  3(1-\delta_0)^2-1-6 \xkh{(1+\delta_0)^2+1}\epsilon - (2+\delta_0) \delta_0.
\]
Note that $\delta_0 \le 1/4$. By taking $\epsilon>0$ sufficiently small we arrive at the conclusion.
\qed


\section{Appendix: Preliminaries and supporting lemmas}
 In this section we shall adopt the following convention. 
\begin{itemize}
\item  For a random variable
$Y$, we shall sometimes use  ``mean" to denote $\mathbb E Y$. This notation is particularly handy when $Y$ is given by a sum of random variables involving
various truncations and modifications. 

\item For a random variable $Y$,  the sub-exponential orlicz norm $\|Y\|_{\psi_1}$ is defined as
$$\|Y\|_{\psi_1} = \inf\{t>0: \mathbb E e^{\frac {|X|} t} \le 2 \}.$$ In particular
$\|Y\|_{\psi_1} \le K \Leftrightarrow \mathbb E e^{\frac {|Y|} K} \le 2$.  Similarly the 
sub-gaussian orlicz norm $\| Y\|_{\psi_2}$ is
\begin{align*}
\|Y\|_{\psi_2} = \inf\{t>0: \mathbb E e^{\frac {|X|^2} t} \le 2 \}.
\end{align*}

\item We denote $\{\va_j\}_{j=1}^m$ as a sequence of i.i.d. random vectors which are
copies of a standard Gaussian random vector $\va:\; \Omega \to \mathbb R^n$ satisfying
$\va\sim \mathcal N(0, I_n)$. 
\end{itemize}

\begin{lemma}[Hoeffding's inequality, \cite{Vershynin2018}] \label{bp0}
Let $X_i$, $1\le i \le m$ be independent, mean zero, sub-gaussian random variables.
Let $b=(b_1,\cdots,b_m) \in \mathbb R^m$.
Then for every $t\ge 0$, we have
\begin{align*}
\mathbb P ( \Bigl| \sum_{i=1}^m b_i X_i \Bigr|\ge t)
\le 2 \exp \Bigl( - \frac {ct^2} { K^2 \|b\|_2^2 } \Bigr),
\end{align*}
where $\max_{i} \| X_i\|_{\psi_2} = K$.
\end{lemma}

\begin{lemma}[Bernstein's inequality, \cite{Vershynin2018}] \label{bp1}
Let $X_i$, $1\le i\le m$ be independent, mean zero, sub-exponential random variables.
Let $b=(b_1,\cdots, b_m) \in \mathbb R^m$. 
Then for every $t\ge 0$, we have
\begin{align*}
\mathbb P ( \Bigl|\frac 1m \sum_{i=1}^m  b_i X_i \Bigr| >t)
\le 2 \exp\left[ -c \cdot \min\bigl(\frac {m^2 t^2} {K^2\|b\|_2^2}, \frac {mt}
{K\|b\|_{\infty} }\bigr) \right],
\end{align*}
where $ \max_i \| X_i \|_{\psi_1} = K$. 
\end{lemma}

\begin{lemma} \label{bp4}
For any $\epsilon>0$, there exists $N_0=N_0(\epsilon)>0$, such that for any $N\ge N_0$,
we have
\begin{align*}
\frac 1m \sum_{j=1}^m \mathbb E
\left( (\va_j^\T \vu)^2 (\va_j^\T \vv)^2 \chi_{|\va_j^\T \vv| \ge N} \right) \le \epsilon,
\end{align*}
where $\vu\in \mathbb S^{n-1}$, $\vv\in \mathbb S^{n-1}$.

\end{lemma}

{\noindent\it \textbf{Proof}}~
Since $\va_j$ are i.i.d., it suffices to prove the statement for a single random vector $\va
\sim \mathcal N(0,I_n)$. Noting that $\va^\T \vu \sim \mathcal N(0,1)$ and $\va^\T \vv \sim \mathcal 
N(0,1)$, we have
\begin{align*}
&\mathbb E\left(  (\va^\T \vu)^2 (\va^\T \vv)^2 \chi_{|\va^\T \vv|\ge N}  \right) \notag \\
\le &\; 
\frac 14 \epsilon \cdot   \mathbb E (\va^\T \vu)^4 + \frac 1{\epsilon} \cdot
\mathbb E ( (\va^\T \vv)^4 \chi_{|\va^\T \vv|\ge N }) \notag \\
\le & \frac 34 \epsilon+ \frac 1 {\epsilon}  
\cdot \frac 1 {\sqrt{2\pi}} \int_{|y|\ge N}
e^{-\frac {y^2}2 } y^4 dy \le \epsilon,
\end{align*}
if $N$ is sufficiently large.  Note that one can easily quantify $N_0$ in terms of $\epsilon$.
However we shall not dwell on this here.
\qed

\begin{lemma} \label{bp5}
Let $A=\frac 1m \sum_{j=1}^m \va_j \va_j^\T $. For any $0<\epsilon \le 1$,  if $m\ge C \epsilon^{-2} n$ then
\begin{align*}
\mathbb P ( \| A- \operatorname{I} \|_{2} >\epsilon)
\le \exp(- c m\epsilon^2).
\end{align*}
In particular, for $m\ge C \epsilon^{-2} n$, with probability at least
$1-\exp(-c \epsilon^2 m)$, we have
\begin{align*}
\Bigl | \frac 1 m\sum_{j=1}^m (\va_j^\T \vu) (\va_j^\T \vv) - 
\operatorname{mean} \Bigr| \le \epsilon, \qquad\forall\,
\vu, \vv \in \mathbb S^{n-1}.
\end{align*}
Here, $C$ and $c$ are universal positive constants.
\end{lemma}

{\noindent\it \textbf{Proof}}~
We briefly sketch the standard proof here for the sake of completeness. By using a 
$\delta$-net $S_{\delta}$ on $\mathbb S^{n-1}$  with $0<\delta<\frac 12$ and
$\operatorname{Card}(
S_{\delta})\le (1+\frac 2 {\delta})^n$, we have
\begin{align*}
\| A- \operatorname{I} \|_{\operatorname{op}}
\le \frac 1 {1-2\delta} \sup_{\vx, \vy \in S_{\delta}} 
\langle (A-\operatorname{I}) \vx, \vy \rangle.
\end{align*}

Now for a pair of fixed $\vx_0$, $\vy_0 \in S_{\delta}$, since $\| (\va_j^\T \vx) (\va_j^\T \vy) \|_{\psi_1}
\le \| \va_j^\T \vx_0\|_{\psi_2} \| \va_j^\T \vy_0\|_{\psi_2} \lesssim 1$, by using
Lemma \ref{bp1}, we have for any $0<\epsilon_1\le 1$,
\begin{align*}
\mathbb P ( \langle (A-\operatorname{I} ) \vx_0, \vy_0 \rangle  > \epsilon_1)
\le 2 \exp(-c' \cdot m \cdot \epsilon_1^2).
\end{align*}
For any $0< \epsilon<1$,  taking $\epsilon_1:=(1-2\delta) \epsilon$ and $\delta=\frac 14$, we have 
\begin{align*}
\mathbb P ( \| A- \operatorname{I} \|_{2 } >\epsilon)
\le 2(1+\frac 2 {\delta})^{2n} \exp
\bigl[ - c' m (1-2\delta)^2 \epsilon^2 \bigr] \le 2 \exp(-c \epsilon^2 m),
\end{align*}
provided $m\ge C \epsilon^{-2}n$. Here, $C$ and $c$ are universal positive constants.
\qed

\begin{lemma} \label{bp6}
Suppose $h:\, \mathbb R\to \mathbb R$ is a locally Lipschitz continuous function such that
\begin{align*}
|h(z) - h(\tilde z)| \lesssim (1+|z|+|\tilde z|) |z-\tilde z|, \qquad \forall\, z,\tilde z \in \mathbb R.
\end{align*}
Assume that  $\| h( Z) \|_{\psi_1} \lesssim 1$ for a standard Gaussian random variable $Z\sim N(0,1)$.
For any $0<\epsilon \le \frac 12$, if $m\ge  C  \epsilon^{-2}  \log(1/\epsilon) n$, then 
with probability at least $1- 3\exp(-c \epsilon^2 m)$, it holds
\begin{align*}
\Bigl|
\frac 1m \sum_{j=1}^m h(\va_j^\T \vu) - \operatorname{mean}
\Bigr| \le \epsilon, \qquad \forall\, \vu \in \mathbb S^{n-1}.
\end{align*}
Here, $C$ and $c$ are universal positive constants.
\end{lemma}

{\noindent\it \textbf{Proof}}~
Introduce a $\delta$-net $S_{\delta}$ on $\mathbb S^{n-1}$ with
$\operatorname{Card}(S_{\delta}) \le (1+\frac 2 {\delta} )^n$. Observe that
by Lemma \ref{bp1}, for any $0<\epsilon_1\le 1$,  it holds
\begin{align*}
\mathbb P ( \sup_{u \in S_{\delta} }\Bigl| \frac 1m \sum_{j=1}^m h(\va_j^\T \vu)-
\operatorname{mean} \Bigr|  > \epsilon_1)
\le (1+\frac 2 {\delta} )^n \cdot 2 \cdot \exp( -c_1 \cdot m\cdot \epsilon_1^2).
\end{align*}
By Lemma \ref{bp5}, we have for $m\ge Cn$  with probability at least $1-\exp(-c_2 m)$,
it holds that
\begin{align*}
\frac 1m \sum_{j=1}^m |\va_j^\T \vu|^2 \le 2, \quad\forall\, \vu \in \mathbb S^{n-1}.
\end{align*}
 Thus 
with probability at least $1-\exp(-c_2 m)$ and uniformly
for $\vu$, $\vv\in \mathbb S^{n-1}$, we have 
\begin{eqnarray*}
& &\Bigl| \frac 1m \sum_{j=1}^m h(\va_j^\T \vu) - \frac 1m \sum_{j=1}^m h(\va_j^\T \vv) \Bigr| \notag \\
& \lesssim &  \frac 1m \sum_{j=1}^m |\va_j^\T (\vu-\vv)|  
+  \frac 1m \sum_{j=1}^m |\va_j^\T (\vu-\vv)| (|\va_j^\T \vu|+|\va_j^\T \vv|) \notag \\
& \le & \sqrt{ \frac 1m \sum_{j=1}^m |\va_j^\T (\vu-\vv)|^2  } \cdot \xkh{1+  \sqrt{\frac 1m \sum_{j=1}^m |\va_j^\T \vu|^2}  + \sqrt{\frac 1m \sum_{j=1}^m |\va_j^\T \vv|^2}  }\\
& \le & 10 \norm{\vu-\vv} \le 10 \delta.
\end{eqnarray*}
Now we take $\epsilon_1=\frac {\epsilon} 2$ and $\delta= \frac {\epsilon}{20}$.
It follows that for $m\ge  C  \epsilon^{-2} \log(1/\epsilon) n$ with probability at least
\begin{align*}
1- \exp(-c_2 m) - (1+\frac {40} {\epsilon} )^n \cdot 2 \cdot \exp(-c_1 \cdot m \cdot \epsilon^2) \ge 1-3\exp(-c\epsilon^2 m),
\end{align*}
the desired inequality holds uniformly for all $\vu\in \mathbb S^{n-1}$. 
\qed

\begin{corollary} \label{bp6.01}
Let $0<\epsilon \le 1/2$.  Assume $m \ge C \epsilon^{-2} \log(1/\epsilon)  n$.
There exists $N_0(\epsilon)=C \sqrt{\log(1/\epsilon) }>0$ such that for any $N\ge N_0$,
 with probability at least $1- \exp(-c  \epsilon^2 m)$, it holds
\begin{align*}
\frac 1 m \sum_{j=1}^m \chi_{|\va_j^\T \vu|>N} \le \epsilon,
\qquad\forall\, \vu \in \mathbb S^{n-1}.
\end{align*}
\end{corollary}

{\noindent\it \textbf{Proof}}~
We choose $\phi \in C_c^{\infty}(\mathbb R)$ such that $\phi(z)\equiv 1$ for 
$|z|\le \frac 12$, $\phi(z)=0$ for $|z|\ge 1$, and $0\le \phi(z)\le 1$ for all $z\in \mathbb R$. 
Clearly then
\begin{align*}
\frac 1 m \sum_{j=1}^m \chi_{|\va_j^\T \vu|>N} \le
\frac 1m \sum_{j=1}^m (1- \phi( \frac {\va_j^\T \vu} N)) .
\end{align*}
We can then apply Lemma \ref{bp6} with $h(z) =1- \phi(\frac z N)$.  Note that
(below $Z\sim \mathcal N(0,1)$ is a standard normal random variable)
\begin{align*}
\operatorname{mean}= \mathbb E ( 1- \phi (\frac Z N))
\le \mathbb E \chi_{|Z|\ge \frac 12 N} \le O(e^{-N^2}) \le \frac 12\epsilon,
\end{align*}
if $N\ge N_0(\epsilon)$. 
\qed

\begin{lemma} \label{bp7}
Let $X_i$: $1\le i\le m$ be independent random variables with 
\begin{align*}
\max_{1\le i \le m} \mathbb E |X_i|^4 \lesssim 1.
\end{align*}
Then for any $t>0$, 
\begin{align*}
\mathbb P ( \Bigl| \frac 1m \sum_{j=1}^m X_i  - 
\operatorname{mean} \Bigr| > t)  \lesssim \frac 1{m^2 t^4}.
\end{align*}
\end{lemma}

\begin{proof}
Without loss of generality we can assume $X_i$ has zero mean. The result then follows from Markov's inequality and the observation
\begin{align*}
\mathbb E (\sum_{j=1}^m X_j )^4 \lesssim \sum_{i<j} \mathbb E X_i^2 X_j^2 +
\sum_{i} \mathbb E X_i^4 \lesssim m^2.
\end{align*}
\end{proof}

\begin{lemma} \label{bp8}
Let $0<\epsilon \le  1/2$.  Assume $m \ge C \epsilon^{-4} \log(1/\epsilon) n$.
There exists $N_0(\epsilon)=C \sqrt{\log(1/\epsilon) }>0$  such that for any $N\ge N_0$,  with probability
at least $1- \frac {c' }{m^2}- \exp(-c'' \epsilon^2 m)$, we have
\begin{align*}
\frac 1m
\sum_{j=1}^m (\va_j^\T \vx)^4 \chi_{|\va_j^\T \vu |\ge N} \le \epsilon,
\qquad\forall\, \vu \in \mathbb S^{n-1}.
\end{align*}
Here, $C, c'$ and $c''$ are universal positive constants.
\end{lemma}
\begin{proof}
By using Cauchy-Schwartz, we have
\[
\frac 1m \sum_{j=1}^m (\va_j^\T \vx)^4 \chi_{|\va_j^\T \vu |\ge N} \le  \sqrt{\frac 1m  \sum_{j=1}^m (\va_j^\T \vx)^8} \sqrt{ \frac 1 { m} \sum_{j=1}^m \chi_{|\va_j^\T \vu|\ge N }}.
\]
By using Lemma \ref{bp7}, we have for any $t_1>0$,
\begin{align*}
\mathbb P
\Bigl( \Bigl| \frac 1m \sum_{j=1}^m (\va_j^\T \vx)^8 -
\operatorname{mean } \Bigr|> t_1 \Bigr)
\lesssim \frac 1 {m^2 t_1^4}.
\end{align*}
Choosing $t_1$ to be an absolute constant.   The desired result then follows from Corollary \ref{bp6.01}.
\end{proof}

\begin{lemma} \label{bp8.02}
For any $0<\epsilon \le  1/2$,  there exist constants $C(\epsilon), c(\epsilon)>0$ only depending on $\epsilon$  such that if $m\ge C(\epsilon)  n$, then the following holds with probability at least $1-3\exp(-c(\epsilon) m)$:
\begin{align*}
\frac 1m \sum_{j=1}^m (\va_j^\T \vu)^4 \ge 3 -\epsilon, \qquad\forall\, \vu \in \mathbb S^{n-1}.
\end{align*}
\end{lemma}
\begin{proof}
Let $\phi \in C_c^{\infty}(\mathbb R)$ be such that $0\le \phi(x) \le 1$ for all $x \in \mathbb R$,
$\phi(x) =1$ for $|x| \le 1$ and $\phi(x)=0$ for $|x|\ge 2$. Then  for any $N\ge 1$, 
\begin{align*}
\frac 1m \sum_{j=1}^m (\va_j^\T \vu)^4 \ge \frac 1m
\sum_{j=1}^m (\va_j^\T \vu)^4 \phi(\frac {\va_j^\T \vu} N).
\end{align*}
We first take $N$ sufficiently large (depending only on $\epsilon$) such that
\begin{align*}
\frac 1m \sum_{j=1}^m\mathbb E (\va_j^\T \vu)^4 \phi(\frac {\va_j^\T \vu} N)
=\mathbb E Z^4 \phi(\frac {Z} N) \ge 3-\frac {\epsilon}2,
\end{align*}
where $Z \sim \mathcal N(0,1)$. 
Then by using Lemma \ref{bp6}, we have for $m\ge C(\epsilon) \cdot n$, 
\begin{align*}
\mathbb P ( \Bigl| \frac 1m \sum_{j=1}^m  (\va_j^\T \vu)^4 \phi(\frac {\va_j^\T \vu} N)
-\operatorname{mean} \Bigr|>\frac {\epsilon}2 ) \le  3 \exp(- c(\epsilon) m).
\end{align*}
The desired result then easily follows.
\end{proof}

\begin{lemma} \label{bp9}
Suppose $f_1, f_2:\, \mathbb R\to \mathbb R$ are  Lipschitz continuous functions such that
\begin{align*}
&|f_1(z)|+|f_2(z)| \le L \cdot (1+|z|), \qquad\forall\, z \in \mathbb R;\\
&|f_k(z) - f_k(\tilde z)| \le L\cdot |z-\tilde z|, \qquad \forall\, z,\tilde z \in \mathbb R,
\; k=1,2,
\end{align*}
where $L>0$ is a constant. 
Assume that $\| f_1( Z) \|_{\psi_2} +\|f_2(Z)\|_{\psi_2} \lesssim 1$
 for a standard Gaussian random variable $Z\sim N(0,1)$. For any $0<\epsilon \le 1/2$, there exist constant $C_1 >0$ only depending on $L$ and universal constant $c>0$ such that if $m\ge  C_1 \epsilon^{-2} \log(1/\epsilon)  n$, then 
with probability at least $1- \exp(-c \epsilon^2 m)$, we have
\begin{align*}
\Bigl|
\frac 1m \sum_{j=1}^m f_1(\va_j^\T \vu)f_2(\va_j^\T \vv) - \operatorname{mean}
\Bigr| \le \epsilon, \qquad \forall\, \vu,\, \vv \in \mathbb S^{n-1}.
\end{align*}
\end{lemma}

\begin{proof}
Introduce a $\delta$-net $S_{\delta}$ on $\mathbb S^{n-1}$ with
$\operatorname{Card}(S_{\delta}) \le (1+\frac 2 {\delta} )^n$. 
By Lemma \ref{bp1}, for any $0<\epsilon_1\le 1$,  we have
\begin{align*}
\mathbb P ( \sup_{\vu,\vv \in S_{\delta} }\Bigl| \frac 1m \sum_{j=1}^m f_1(\va_j^\T \vu)
f_2(\va_j^\T \vv)-
\operatorname{mean} \Bigr|  > \epsilon_1)
\le (1+\frac 2 {\delta} )^{2n} \cdot 2 \cdot \exp( - c_1 \epsilon_1^2  m),
\end{align*}
where $ c_1>0$ is a universal constant. 
Next by Lemma \ref{bp5} and \ref{bp6},  if $m\ge C n$ then  with probability at least $1-\exp(-c_2 m)$, it holds 
\begin{align*}
&\Bigl| \frac 1 m \sum_{k=1}^m |\va_j^\T \vw| - \operatorname{mean}
\Bigr| \le 0.01, \qquad\forall\, \vw \in \mathbb S^{n-1}; \\
&\Bigl| \frac 1 m \sum_{k=1}^m |\va_j^\T \vw|^2 - \operatorname{mean}
\Bigr| \le 0.01, \qquad\forall\, \vw \in \mathbb S^{n-1},
\end{align*}
where $C,c_2$ are universal positive constants.
Consequently,  for any $\vu, \vv\in \mathbb S^{n-1}$, there exist $\tilde \vu, \tilde \vv \in S_{\delta}$ such that $\|\vu-\tilde \vu\|_2\le \delta$, $\|\vv-\tilde \vv\|_2 \le \delta$, and then
\begin{align*}
&\frac 1m \sum_{j=1}^m |f_1(\va_j^\T \vu) f_2(\va_j^\T \vv)
-f_1(\va_j^\T \tilde \vu) f_2(\va_j^\T \tilde \vv) | \notag \\
\le & \;
\frac 1m \sum_{j=1}^m |f_1(\va_j^\T \vu) -f_1(\va_j^\T \tilde \vu) | |f_2(\va_j^\T \vv)|
+ \frac 1m \sum_{j=1}^m |f_1(\va_j^\T \tilde \vu)| |f_2(\va_j^\T \vv) -f_2(\va_j^\T \tilde \vv) | \notag \\
\le & \; 
\frac 1m \sum_{j=1}^m L^2 |\va_j^\T (\vu-\tilde \vu) |  (1+|\va_j^\T \vv|)
+ \frac 1m \sum_{j=1}^m L^2 (1+|\va_j^\T \tilde \vu|) |\va_j^\T (\vv-\tilde \vv) |\notag \\
\le & \; \frac {L^2}m \sum_{j=1}^m ( |\va_j^\T (\vu-\tilde \vu) | +|\va_j^\T (\vv-\tilde \vv)|)+
\frac {L^2}m \sum_{j=1}^m \frac 1{4{\delta} }
|\va_j^\T (\vu-\tilde \vu)|^2 +
\frac {L^2}m \sum_{j=1}^m \delta |\va_j^\T \vv|^2  \notag \\
&\quad+\frac {L^2}m \sum_{j=1}^m  
\frac 1 {4\delta} |\va_j^\T (\vv-\tilde \vv)|^2 +
\frac {L^2}m \sum_{j=1}^m \delta |\va_j^\T \tilde \vu|^2 \notag \\
\le & \; 3 L^2 (\|\vu-\tilde \vu\|_2 + \|\vv-\tilde \vv\|_2) +
\frac 1 {\delta} L^2(\|\vu-\tilde \vu\|_2^2 +\|\vv-\tilde \vv\|_2^2)+ 2\delta L^2
\le 10 \delta L^2.
\end{align*}
Now set $\delta=\frac {\epsilon}{15 L^2}$ and $\epsilon_1=\frac{\epsilon}3$. 
The desired conclusion then follows  with probability at least
\begin{align*}
1- \exp(-c_2 m)-
(1+\frac 2 {\delta} )^{2n} \cdot 2 \cdot \exp( -c_1 \epsilon_1^2 m)
\ge 1-\exp(-c \epsilon^2 m),
\end{align*}
provided $m\ge C_1 \epsilon^{-2} \log(1/\epsilon) n$. Here, $c$ is a  positive universal constant and 
$C$ is a positive constant only depending on $L$.
\end{proof}

\begin{corollary} \label{bp9.01}
If $m\ge Cn$, then with probability at least $1- \exp(-cm)$, we have
\begin{align*}
\frac 1m \sum_{j=1}^m (\va_j^\T \vu)^2 (\va_j^\T \vv)^2 
\ge c_1>0, \qquad \forall\, \vu, \vv \in \mathbb S^{n-1},
\end{align*}
where $C, c$ and $c_1$ are absolute positive constants.  
\end{corollary}

\begin{proof}
Step 1. Write $\vv= s \vu +\sqrt{1-s^2} \vu^{\perp}$, where $|s|\le 1$, and $\vu^{\perp} \in \mathbb S^{n-1}$ is such that $\nj{\vu^{\perp}, \vu}=0$.  Let $\va\sim \mathcal N(0,I_n)$ and denote
$X=\va^\T \vu$, $Y=\va^\T \vu^{\perp}$. Clearly $X\sim \mathcal N(0,1)$, $Y \sim \mathcal N(0,1)$
and $X$, $Y$ are independent.  Now let $N\ge 4$. We have
\begin{align*}
 & \mathbb E (\va^\T \vu)^2 (\va^\T \vv)^2 \chi_{|\va^\T \vu|\le N} \chi_{|\va^\T \vv|\le N} \notag \\
 = &\; \mathbb E  X^2 ( s X +\sqrt{1-s^2} Y)^2 \chi_{|X|\le N}
 \chi_{|s X+\sqrt{1-s^2} Y| \le N} \notag \\
 \ge &\; \mathbb E X^2 (sX+\sqrt{1-s^2} Y)^2 \chi_{|X|\le \frac N4}
 \chi_{|Y| \le \frac N4} \notag \\
 = &\; \mathbb E (s^2 X^4 + (1-s^2) X^2 Y^2) \chi_{|X|\le \frac N4}
 \chi_{|Y|\le \frac N4}  \ge 2c_1>0,
 \end{align*}
where $c_1>0$ is an absolute constant, and $N$ is taken to be a sufficiently large
absolute constant. 

Step 2. Let $\phi \in C_c^{\infty}(\mathbb R)$ be such that $\phi(x) =x^2$
for $|x| \le N$ and $\phi(x)=0$ for $|x|\ge N+1$. Clearly
if $m \ge Cn$, then with probability at least $1-\exp(-cm)$, we have
\begin{align*}
\Bigl|\frac 1m \sum_{j=1}^m \phi(\va_j^\T \vu) \phi(\va_j^\T \vv)
-\operatorname{mean} \Bigr| \le \frac 12c_1
\end{align*}
and thus
\begin{align*}
 &\frac 1 m \sum_{j=1}^m (\va_j^\T \vu)^2 (\va_j^\T \vv)^2 
\ge\; \frac 1m \sum_{j=1}^m \phi(\va_j^\T \vu) \phi(\va_j^\T \vv)  >c_1,
\quad\forall\, \vu,\vv \in \mathbb S^{n-1}.
\end{align*}
\end{proof}

\begin{lemma} \label{bp10}
Let $h:\, \mathbb R \to \mathbb R$ be a Lipschitz continuous function such that
\begin{align*}
& \sup_{z\in \mathbb R} \frac {|h(z)|} {1+|z|} \lesssim 1 \quad \mbox{and} \quad \sup_{z\ne \tilde z} 
\frac {|h(z)-h(\tilde z) | } { |z-\tilde z|  } \lesssim 1.
\end{align*}
Define the set
\begin{align*}
F = \{ \vu \in \mathbb R^n:\; \|\vu\|_2 \le 1, \quad \vu^\T \vx=0 \}.
\end{align*}
Suppose $\vb=(b_1,\cdots,b_m) \in \mathbb R^m$ satisfies $\norm{\vb} \lesssim \sqrt{m} $. For any $0<\epsilon\le  1/2$, if $m\ge C \epsilon^{-2} \log (1/\epsilon) n$, then
with probability at least $1-\exp(-c \epsilon^2 m)$, we have
\begin{align*}
\Bigl| 
\frac 1m \sum_{j=1}^m b_j h(\va_j^\T \vu) - \operatorname{mean}
\Bigr| \le \epsilon, \qquad \forall\, \vu \in F.
\end{align*}

\end{lemma}

\begin{proof}
First it is easy to check that $\max_{j} \| h(\va_j^\T \vu) \|_{\psi_2} \lesssim 1$. 
By Lemma \ref{bp0}, for each $\vu \in F$, we have
\begin{align*}
\mathbb P ( 
\Bigl| \frac 1m \sum_{j=1}^m b_j h(\va_j^\T \vu) - \operatorname{mean}
\Bigr| > \frac{\epsilon}4)
\le 2 \exp( - c m\epsilon^2).
\end{align*}
Now let $\delta>0$ and introduce a $\delta$-net $S_{\delta}$ on the set 
$F$. Note that the set $F$ can be identified as a unit ball in $\mathbb R^{n-1}$.
We have $\operatorname{Card} (S_{\delta} ) \le (1+\frac 2{\delta})^n$. Thus
\begin{align*}
\mathbb P ( 
 \sup_{u  \in S_{\delta} }\Bigl| \frac 1m \sum_{j=1}^m b_j h(\va_j^\T \vu) - \operatorname{mean}
\Bigr| > \frac{\epsilon}4)
\le 2  (1+\frac 2 {\delta} )^n \exp( - c m\epsilon^2).
\end{align*}
By Lemma \ref{bp5}, if $m\ge Cn$, then with probability at least $1-\exp(-cm)$, we have
\begin{align*}
\frac 1m \sum_{j=1}^m (\va_j^\T \vu)^2 \le 2 \|\vu\|_2^2, \qquad \forall\, \vu \in \mathbb R^n.
\end{align*}
Now if $\vu \in S_{\delta}$, $\vv\in F$ with $\|\vv-\vu\|_2\le \delta$, then with
probability at least $1-\exp(-cm)$, we have
\begin{align*}
 & \Bigl| \frac 1m \sum_{j=1}^m b_j  \cdot h(\va_j^\T \vu) 
    -\frac 1m \sum_{j=1}^m b_j 
 \cdot  h(\va_j^\T \vv)   \Bigr|\notag \\
  \le&\;  \frac 1m \sum_{j=1}^m 
 |b_j| \cdot K_0 \cdot |\va_j^\T (\vu-\vv) |   \notag \\
 \le & K_0 \sqrt{\frac 1m \sum_{j=1}^m b_j^2}  \sqrt{\frac 1m \sum_{j=1}^m |\va_j^\T (\vu-\vv)|^2}\notag \\
 \le & K_1 \| \vu-\vv\|_2,
\end{align*}
where $K_0>0$, $K_1>0$ are absolute constants.  On the other hand
\begin{align*}
& \frac 1m \sum_{j=1}^m |b_j| \mathbb E | h(\va_j^\T \vu) -h(\va_j^\T \vv) | \notag \\
\lesssim &\; \frac 1m \sum_{j=1}^m |b_j| \cdot  \norm{\vu-\vv} \lesssim \norm{\vu-\vv}. 
\end{align*}
 It follows that for some absolute constant $K_3>0$, 
\begin{align*}
 & \Bigl| \frac 1m \sum_{j=1}^m b_j  \cdot \bigl(  h(\va_j^\T \vu)  -\mathbb E (h(\va_j^\T \vu) ) \bigr)
    -\frac 1m \sum_{j=1}^m b_j 
 \cdot \bigl( h(\va_j^\T \vv)  -\mathbb E( h(\va_j^\T \vv) ) \bigr)   \Bigr|\notag \\
\le & \; K_3 \norm{\vu-\vv}.
\end{align*}
Now set $\delta=\frac {\epsilon}{10K_3+10}$.  The desired conclusion then
follows with probability at least
\begin{align*}
1-2(1+\frac 2{\delta})^n \exp(-cm \epsilon^2) -\exp(-cm)
\ge 1- \exp(-c_1 \epsilon^2 m),
\end{align*}
provided $m\ge C  \epsilon^{-2} \log(1/\epsilon) n$, where $c_1>0$ is an absolute constant.
\end{proof}

\begin{lemma} \label{bp11}
Let $h:\, \mathbb R \to \mathbb R$ be a Lipschitz continuous function such that
\begin{align*}
& \sup_{z\in \mathbb R} \frac {|h(z)|} {1+|z|} \lesssim 1 \quad \mbox{and} \quad \sup_{z\ne \tilde z} 
\frac {|h(z)-h(\tilde z) | } { |z-\tilde z|  } \lesssim 1.
\end{align*}
Let $f_3:\, \mathbb R \to \mathbb R$ be such that
\begin{align*}
\sup_{z\in \mathbb R} \frac  { f_3(z)} {1+|z|^3} \lesssim 1.
\end{align*}
Define 
\begin{align*}
F = \{ \vu \in \mathbb R^n:\; \|\vu\|_2 \le 1, \quad \vu^\T \vx=0 \}.
\end{align*}
Then for any $0<\epsilon \le  1/2$,  there exist $C_1=C_1(\epsilon)>0$, $C_2=C_2(\epsilon)>0$,
 such that if $m\ge C_1 \cdot n$, then
the following holds with probability at least $1-\frac {C_2} {m^2}$:
\begin{align*}
\Bigl| \frac 1m \sum_{j=1}^m f_3(\va_j^\T \vx) h(\va_j^\T \vu)
- \operatorname{mean} \Bigr|\le \epsilon, \qquad\forall\, \vu \in F.
\end{align*}
\end{lemma}

\begin{proof}
Step 1. Set $b_j = f_3(\va_j^\T \vx)$.  Denote 
\begin{align*}
K_0= \mathbb E b_1^2 \lesssim 1.
\end{align*}
By  Lemma \ref{bp7}, we have
\begin{align*}
\mathbb P ( \Bigl| \frac 1m \sum_{j=1}^m b_j^2 
-K_0 \Bigr| >t ) \lesssim \frac 1 {m^2 t^4}.
\end{align*}
Then with
probability at least $1- \frac 1 {m^2 }$, we have
\begin{align*}
\frac 1m \sum_{j=1}^m b_j^2 \le B_0,
\end{align*}
where $B_0>0$ is some absolute constant. 

Step 2. Denote $\tilde \va_j = \va_j -(\va_j^\T \vx) \vx$.  An important observation is that 
$(\va_j^\T \vx)_{1\le j\le m}$ and $(\tilde \va_j)_{1\le j\le m}$ are independent. 
Note that for $\vu \in F$ we have $\va_j^\T \vu=\tilde \va_j^\T \vu$.  Thus for every $\tilde b_j$ with the property $\frac 1m \sum_{j=1}^m (\tilde b_j)^2 \le B_0$, we have the following as
a consequence of Lemma \ref{bp10}:
For any $0<\epsilon \le  1/2$, if $m\ge C  \epsilon^{-2} \log(1/\epsilon) n$, then
with probability at least $1-\exp(-c \epsilon^2 m)$, we have
\begin{align*}
\Bigl| \frac 1m \sum_{j=1}^m \tilde b_j\cdot  \bigl(h(\tilde \va_j^\T \vu) - \mathbb E (h(\tilde \va_j^\T \vu ) ) \bigr)\Bigr| \le \frac 13  \epsilon, \qquad\forall\, \vu \in F.
\end{align*}

Step 3.  By using the results from Step 1 and Step 2, with probability at least  $1- \frac 2 {m^2 }$, we have
\begin{align*}
\Bigl| \frac 1m \sum_{j=1}^m  b_j \cdot  \bigl(h(\tilde \va_j^\T \vu) - \mathbb E (h(\tilde \va_j^\T \vu ) ) \bigr)\Bigr| \le  \frac 13\epsilon, \qquad\forall\, \vu \in F.
\end{align*}
Now note that $|\mathbb E ( h(\tilde \va_j^\T \vu) )| =| \mathbb E (h (\tilde \va_1^\T \vu) )|
\lesssim 1$.  By Lemma \ref{bp7}, we have
\begin{align*}
\mathbb P ( \Bigl| \frac 1m \sum_{j=1}^m  (b_j -\mathbb E b_j ) \Bigr |>t)
\lesssim  \frac 1 {m^2  t^4}. 
\end{align*}
Choosing $t= \frac {\epsilon}{K_1}$ where $K_1>0$ is a sufficiently large absolute constant
such that
\begin{align*}
 \frac 1 {K_1} \cdot |\mathbb E b_1| \le \frac {\epsilon}{10}
 \end{align*}
then yields the result.
\end{proof}

\begin{corollary} \label{bp11.01}
For any $0<\epsilon \le  1/2$,  there exist $C_1=C_1(\epsilon)>0$, $C_2=C_2(\epsilon)>0$,
 such that if $m\ge C_1 \cdot n$, then
the following holds with probability at least $1-\frac {C_2} {m^2}$:
\begin{align*}
\Bigl| \frac 1m \sum_{j=1}^m (\va_j^\T \vx)^3 (\va_j^\T \vu)
- \operatorname{mean} \Bigr|\le \epsilon, \qquad\forall\, \vu \in \mathbb S^{n-1}.
\end{align*}
\end{corollary}
\begin{proof}
We decompose $\vu=(\vu^\T\vx)\vx+ \tilde \vu$, where $\tilde \vu^\T\vx=0$. 
The result then easily follows from Lemma \ref{bp11}.
\end{proof}

\begin{lemma} \label{bp12}
For any $0<\epsilon \le  1/2 $, there exists $N_0=N_0(\epsilon)>0$,
$C_1=C_1(\epsilon)>0$, $C_2=C_2(\epsilon)>0$, such that if $m\ge C_1 n$,
then the following hold with probability at least $1-\frac {C_2} {m^2}$:
For any $N\ge N_0$, we have
\begin{align*}
&\frac 1m \sum_{j=1}^m |\va_j^\T \vx|^3 |\va_j^\T \vu|
\chi_{|\va_j^\T \vx| \ge N} \le \epsilon, \qquad\forall\, \vu \in \mathbb S^{n-1};\\
&\frac 1m \sum_{j=1}^m |\va_j^\T \vx|^3 |\va_j^\T \vu|
\chi_{|\va_j^\T \vu| \ge N} \le \epsilon, \qquad\forall\, \vu \in \mathbb S^{n-1}.
\end{align*}
\end{lemma}

\begin{proof}
We only sketch the proof.  Write $\vu=(\vu^\T\vx)\vx+ \tilde \vu$, where $\tilde \vu
\in F= \{ \tilde \vu\in \mathbb R^n:\; \|\tilde \vu\|_2\le 1,  \tilde \vu^\T \vx=0\}$. 
For the first inequality, note that (observe $|\vu^\T\vx|\le 1$)
\begin{align*}
&\frac 1m \sum_{j=1}^m |\va_j^\T \vx|^3 |\va_j^\T \vu| \chi_{|\va_j^\T \vx|\ge N} \notag \\
\le &\frac 1m \sum_{j=1}^m |\va_j^\T \vx|^4 \chi_{|\va_j^\T \vx |\ge N} 
+\frac 1 m\sum_{j=1}^m |\va_j^\T \vx|^3 |a_j \cdot \tilde \vu| \chi_{|\va_j^\T \vx| \ge N}.
\end{align*}
For the first term one can use Lemma \ref{bp7}. For the second term one
can use  Lemma \ref{bp11}. 

Now for the second inequality, we write
\begin{align*} 
&\frac 1m \sum_{j=1}^m |\va_j^\T \vx|^3 |\va_j^\T \vu|
\chi_{|\va_j^\T \vu| \ge N}  \notag \\
\le &\; \underbrace{\frac 1m \sum_{j=1}^m |\va_j^\T \vx|^3 \chi_{|\va_j^\T \vx| \le M} 
|\va_j^\T \vu| \chi_{|\va_j^\T \vu| \ge N}}_{=:H_1}
+ \underbrace{\frac 1m \sum_{j=1}^m |\va_j^\T \vx|^3 \chi_{|\va_j^\T \vx|>M}
|\va_j^\T \vu|}_{=:H_2}.
\end{align*}
For $H_2$,  by using the estimates already obtained in the beginning part of this proof,
it is clear that we can take $M$ sufficiently large such that $H_2\le \epsilon/2$. 
After $M$ is fixed, we return to the estimate of $H_1$. 
Note that we can work with a smoothed cut-off function instead of the strict cut-off. 
 The result then follows from
Lemma \ref{bp9} by taking $N$ sufficiently large.

\end{proof}

\begin{lemma} \label{bp13}
Define $ F = \{ \vu \in \mathbb R^n:\; \|\vu\|_2 \le 1, \quad \vu^\T \vx=0 \}$.
Suppose $\vb=(b_1,\cdots,b_m) \in \mathbb R^m$ satisfies  $\norm{\vb} \lesssim \sqrt{m} $ and $\norms{\vb}_\infty \lesssim \log m$.
For any $0<\epsilon\le 1/2$,  if $m\ge C \epsilon^{-1} n \log n$, then
with probability at least $1- \exp(-c \epsilon m/\log m )$, it holds that 
\begin{align*}
\Bigl| 
\frac 1m \sum_{j=1}^m b_j  (\va_j^\T \vu)(\va_j^\T \vv) - \operatorname{mean}
\Bigr| \le \epsilon, \qquad \forall\, \vu, \vv\in F.
\end{align*}

\end{lemma}

\begin{proof}
 Let $0<\delta <  1/2$ and introduce a $\delta$-net $S_{\delta}$ on the set 
$F$. Note that the set $F$ can be identified as a unit ball in $\mathbb R^{n-1}$.
We have $\operatorname{Card} (S_{\delta} ) \le (1+\frac 2{\delta})^n$. Introduce the operator
\begin{align*}
A = \frac 1m \sum_{j=1}^m b_j (\va_j \va_j^T -\operatorname{I}).
\end{align*}
We have 
\begin{align*}
\| A \|_{2} =\sup_{\vx,\vy \in F}
\langle A\vx, \vy \rangle \le \frac 1 {1-2\delta}
\sup_{\vx, \vy \in S_{\delta}} \langle A \vx, \vy\rangle.
\end{align*}
By Lemma \ref{bp1}, for each $\vu, \vv \in F$, we have
\begin{align*}
\mathbb P ( 
\Bigl| \frac 1m \sum_{j=1}^m b_j (\va_j^\T \vu)(\va_j^\T \vv)- \operatorname{mean}
\Bigr| > \frac{\epsilon}4)
\le 2 \exp( - c   \min\{  {m \epsilon^2},  \frac {m\epsilon}{\log m } \} )
\le 2 \exp( -c \frac{m}{\log m} \epsilon). 
\end{align*}
Thus
\begin{align*}
\mathbb P ( 
 \sup_{\vu,\vv   \in S_{\delta} }\Bigl| \frac 1m \sum_{j=1}^m b_j (\va_j^\T \vu)(\va_j^\T \vv) - \operatorname{mean}
\Bigr| > \frac{\epsilon}4)
\le 2  (1+\frac 2 {\delta} )^n \exp( - c \frac m{\log m} \epsilon).
\end{align*}
Taking $\delta=\frac 14$ and $m\ge C\epsilon^{-1} n \log m$ then yields the result. 
\end{proof}

\begin{corollary} \label{bp13.01}
Suppose $h_1$, $h_2:\, \mathbb R\to \mathbb R$ are locally Lipschitz continuous
functions such that
\begin{align*}
&\sup_{z \in \mathbb R} \frac {|h_1(z)|+|h_2(z)|} {1+|z|} \lesssim 1, \\
& \sup_{z\ne \tilde z} \frac {|h_i(z)-h_i(\tilde z)|}{|z-\tilde z|} \lesssim 1,
\quad i=1,2.
\end{align*}
Define $ F = \{ \vu \in \mathbb R^n:\; \|\vu\|_2 \le 1, \quad \vu^\T \vx=0 \}$.  Suppose $\vb=(b_1,\cdots,b_m) \in \mathbb R^m$ satisfies  $\norm{\vb} \lesssim \sqrt{m} $ and $\norms{\vb}_\infty \lesssim \log m$.
For any $0<\epsilon\le 1/2$,  there exist a constant $C_1=C_1(\epsilon)>0$, such that  if $m\ge C_1 n \log n$, then with probability at least  $1- \exp(-c \epsilon m/\log m )$, it holds that 
\begin{align*}
\Bigl| 
\frac 1m \sum_{j=1}^m b_j  h_1 (\va_j^\T \vu) h_2(\va_j^\T \vv) - \operatorname{mean}
\Bigr| \le \epsilon, \qquad \forall\, \vu, \vv\in F.
\end{align*}

\end{corollary}

\begin{proof}
 Let $0<\delta < \frac 12$ and introduce a $\delta$-net $S_{\delta}$  with
 $\operatorname{Card} (S_{\delta} ) \le (1+\frac 2{\delta})^n$ on the set 
$F$. As we shall see momentarily, we will need to take $\delta =O(\epsilon)$. 
By Lemma \ref{bp1},  we have
\begin{align*}
\mathbb P ( 
 \sup_{\vu,\vv   \in S_{\delta} }\Bigl| \frac 1m \sum_{j=1}^m b_j h_1 (\va_j^\T \vu)
 h_2 (\va_j^\T \vv) - \operatorname{mean}
\Bigr| > \frac{\epsilon}4)
\le 2  (1+\frac 2 {\delta} )^n \exp( - c \frac m{\log m} \epsilon).
\end{align*}
Now for any $\vu,\vv \in F$, $\tilde \vu, \tilde \vv\in S_{\delta}$ with
 $\|\vu-\tilde \vu\|_2 \le \delta$, $\|\vv-\tilde \vv\|_2 \le \delta$, we have
 \begin{align*}
 & \Bigl| h_1(\va_j^\T \vu) h_2 (\va_j^\T \vv) -  h_1(\va_j^\T \tilde \vu)
  h_2 (\va_j^\T \tilde \vv) \Bigr| \notag \\
  \le  & \;  |h_1(\va_j^\T \vu) -h_1(\va_j^\T \tilde \vu) |
  \cdot |h_2(\va_j^\T \vv)| +  |h_1(\va_j^\T \tilde \vu)|
  \cdot | h_2(\va_j^\T \vv)-h_2 (\va_j^\T \tilde \vv) | \notag \\
  \le & \; K_0 |\va_j^\T (\vu-\tilde \vu) | (1+|\va_j^\T \vv|)
  + K_0 (1+|\va_j^\T \tilde \vu|) |\va_j^\T (\vv-\tilde \vv) | \notag \\
  \le &\; K_0 \Bigl( \frac 1{\delta} (\va_j^\T (\vu-\tilde \vu) )^2
  +\delta (\va_j^\T \vv)^2 +2\delta
  + \frac 1 {\delta} (\va_j^\T (\vv-\tilde \vv) )^2 + \delta (\va_j^\T \tilde \vu)^2 \Bigr),
  \end{align*}
  where $K_0$ is an absolute constant.
Now introduce the operator
\begin{align*}
A = \frac 1m \sum_{j=1}^m |b_j| (\va_j \va_j^T -\operatorname{I}).
\end{align*}
By Lemma \ref{bp13}, with probability at least
$1-\exp( -c \frac {m}{\log m} )$, we have
\begin{align*}
\langle A \vx, \vy \rangle \le 1, \qquad\forall\, \vx,\vy \in F.
\end{align*}
Thus with the same probability, we have
\begin{align*}
&\frac 1m \sum_{j=1}^m |b_j| \cdot \Bigl| h_1(\va_j^\T \vu) h_2 (\va_j^\T \vv) -  h_1(\va_j^\T \tilde \vu)
  h_2 (\va_j^\T \tilde \vv) \Bigr| \notag \\
\le &\; K_0 \Bigl( \frac 1{\delta}\frac 1m \sum_{j=1}^m |b_j|  ((\va_j^\T (\vu-\tilde \vu) )^2 
+ (\va_j^\T (\vv-\tilde \vv) )^2 ) \notag \\
& \qquad + \delta \frac 1m \sum_{j=1}^m |b_j | ( (\va_j^\T \vv)^2 +(\va_j^\T \tilde \vu)^2 )  
+ 2\delta \frac 1m \sum_{j=1}^m |b_j|  \Bigr) \notag \\
\le &\; K_1 \delta,
\end{align*}
where $K_1>0$ is another absolute constant.  It is also not difficult to control the differences
in expectation, i.e. for some absolute constant $K_2>0$,
\begin{align*}
&\frac 1m \sum_{j=1}^m |b_j| \cdot \Bigl| \mathbb E h_1(\va_j^\T \vu) h_2 (\va_j^\T \vv) -  \mathbb E h_1(\va_j^\T \tilde \vu)
  h_2 (\va_j^\T \tilde \vv) \Bigr|  \le  K_2 \delta,
  \end{align*}
 Now take $\delta = \frac {\epsilon}{4(K_1+K_2)}$
and the desired result clearly follows by taking $\frac m {\log m} \gtrsim n$. 
\end{proof}

\begin{lemma} \label{bp14}
For any $0<\epsilon\le  1/2$, there are constants $C_1=C_1(\epsilon)>0$,
$C_2=C_2(\epsilon)>0$, such that if $m\ge C_1 n \log n$, then with
probability at least $1- \frac {C_2}{m^2}$, it holds
\begin{align*}
\Bigl| \frac 1m \sum_{j=1}^m (\va_j^\T \vu)^2 (\va_j^\T \vx)^2
-\operatorname{mean} \Bigr| \le \epsilon,
\qquad\forall\, \vu \in \mathbb S^{n-1}.
\end{align*}
\end{lemma}
\begin{proof}
Write $\vu = (\vu^\T \vx) \vx+\vu^\perp$, where $\nj{\vu^\perp, \vx}=0$. Then
\begin{align*}
&\frac 1m \sum_{j=1}^m (\va_j^\T \vu)^2 (\va_j^\T \vx)^2 \notag \\
=&\; (\vu^\T \vx)^4\cdot \frac 1m \sum_{j=1}^m (\va_j^\T \vx)^4  
+ 2(\vu^\T \vx)  \cdot \frac 1m \sum_{j=1}^m (\va_j^\T \vx)^3 \cdot (\va_j^\T \vu^\perp) \notag \\
& \qquad + \frac 1m \sum_{j=1}^m (\va_j^\T \vx)^2 (\va_j^\T \vu^\perp)^2.
\end{align*}
Clearly the first two terms can be easily handled by Lemma \ref{bp7} and
Lemma \ref{bp11} respectively. For these terms we actually only need
$m\ge Cn$.  To handle the last term we need $m\gtrsim n \log n$. The main observation
is that $(\va_j^\T \vx)$ and $(\va_j^\T \vu^\perp)$ are independent. 
Write $b_j =(\va_j^\T \vx)^2$ and observe that with probability at least
$1- O(m^{-2})$, we have
\begin{align*}
\sum_{j=1}^m b_j^2 \le 100 m, \quad  \max_{1\le j \le m} b_j  \le 100 \log m. 
\end{align*} 
For $m\gtrsim n \log n$, by using Lemma \ref{bp13}, it holds with probability
at least $1-O(m^{-2} )- \exp(-c \epsilon  m/\log m)= 1-O(m^{-2})$ that
\begin{align*}
\Big| \frac 1m \sum_{j=1}^m b_j ((\va_j^\T \vu)^2 - \|\vu\|_2^2)  \Bigr|
\le \frac {\epsilon} 2, \qquad\forall\, \vu \in F=\{v\in \mathbb R^n:\,
\|\vv\|_2\le 1, \, \vv^\T \vx=0 \}.
\end{align*}
By Lemma \ref{bp7}, we have with probability $1-O(m^{-2})$, 
\begin{align*}
\|\vu\|_2 \Bigl |\frac 1m \sum_{j=1}^m b_j
-\operatorname{mean} \Bigr|
\le \Bigl |\frac 1m \sum_{j=1}^m b_j
-\operatorname{mean} \Bigr| \le \frac {\epsilon}2.
\end{align*}
The desired result then easily follows.
\end{proof}


%
%



\end{document}